\documentclass[article,aps,pra,11pt,tightenlines]{revtex4-2}
\usepackage{graphicx,epsfig}
\usepackage{amsmath}
\usepackage{amsfonts}
\usepackage{braket}
\usepackage{bbm}
\usepackage[dvipsnames]{xcolor}
\usepackage{algorithm2e} 
\usepackage{hyperref}
\RestyleAlgo{ruled}
\usepackage{amsthm}
\usepackage{overpic}

\makeatletter
\renewcommand{\p@subsection}{}
\renewcommand{\p@subsubsection}{}
\makeatother

\DeclareMathOperator{\Tr}{Tr}
\DeclareMathOperator*{\argmin}{\arg\!\min}
\DeclareMathOperator*{\argmax}{\arg\!\max}
\DeclareMathOperator{\EX}{\mathbb{E}}
\DeclareMathOperator{\VX}{\mathbb{V}}

\newcommand{\hil}{\mathcal{H}} 
\newcommand{\thil}{\widetilde{\mathcal{H}}} 

\newcommand{\tdim}{\tilde{d}} 

\def\<{\langle}
\def\>{\rangle}
\newcommand{\id}{\mathbbm{1}} 

\DeclareMathOperator{\End}{End}

\newtheorem{theorem}{Theorem}
\newtheorem{lemma}[theorem]{Lemma}

\newtheorem{definition}[theorem]{Definition}

\newtheorem{corollary}[theorem]{Corollary}

\newcommand{\REGRET}{\textnormal{Regret}}
\newcommand{\DIST}{\textnormal{Disturbance}}
\newcommand{\ERR}{\textnormal{Err}}

\begin{document}


\title{\large Learning pure quantum states (almost) without regret}

\author{Josep Lumbreras$^1$} 
\email{josep.lumbreras@u.nus.edu}
\author{Mikhail Terekhov$^2$}
\email{mikhail.terekhov@epfl.ch}
\author{Marco Tomamichel$^{1,3}$}
\email{marco.tomamichel@nus.edu.sg}

\affiliation{$^1$Centre for Quantum  Technologies, National University of Singapore, Singapore}
\affiliation{$^2$School of Computer and Communication Sciences, EPFL, Switzerland}
\affiliation{$^3$Department of Electrical and Computer Engineering, National University of Singapore, Singapore}

\date{\today}

\begin{abstract}
    We initiate the study of sample-optimal quantum state tomography with minimal disturbance to the samples. Can we efficiently learn a precise description of a quantum state through sequential measurements of samples while at the same time making sure that the post-measurement state of the samples is only minimally perturbed? Defining regret as the cumulative disturbance of all samples, the challenge is to find a balance between the most informative sequence of measurements on the one hand and measurements incurring minimal regret on the other. Here we answer this question for qubit states by exhibiting a protocol that for pure states achieves maximal precision while incurring a regret that grows only polylogarithmically with the number of samples, a scaling that we show to be optimal.
\end{abstract}

\maketitle

\section{Introduction}

In this work, we approach quantum state tomography from a new angle. Given sequential access to a finite number of samples of a quantum state, our goal is not only to accurately learn a classical description of the state but also to use measurements that disturb the samples as little as possible. Generally these two goals are incompatible, and we are thus interested in tomography algorithms that find an optimal balance between them. We call this setting quantum state tomography with minimal regret.

Minimizing disturbance is important in many real-world scenarios where the samples that we use for tomography are in fact resources for another tasks\,---\,and we thus want to learn the state in a way that is as non-intrusive as possible, ensuring that the post-measurement states remain useful for their intended purpose. An example of this occurs in quantum key distribution, where tomography can be used to keep reference frames aligned during a run but any disturbance due to tomographic measurements will induce bit errors in the correlations used to extract a secret key. Disturbance is also relevant for state-agnostic resource distillation where resourceful states might be destroyed by tomographic measurements but learning the unknown state is crucial since optimal extraction protocols generally depend on its description. 

In both cases we encounter a fundamental trade-off between exploration (learning the state) and exploitation (using the samples for another purpose). These types of trade-offs are fundamental to the study of adaptive algorithms in machine learning, and our work establishes a strong link between quantum tomography and the classical multi-armed bandit model in reinforcement learning. In fact, one of the main technical ingredients of the present work is a classical bandit algorithm by some of us~\cite{pmlr-v247-lumbreras24a}, originally inspired by shot noise in quantum mechanics.

To illustrate this connection from a physics perspective, it helps to reflect on how measurements disturb quantum systems. A defining feature of quantum theory is that measurements generally disturb the system being measured. But what does this disturbance intuitively mean? To illustrate this, consider a qubit prepared in the pure state

\begin{align}
|\psi\rangle = \sqrt{1 - \epsilon^2} |0\rangle + \epsilon |1\rangle,
\end{align}
with $\epsilon \in [0, 1]$. A projective measurement in the computational basis $\{ |0\rangle, |1\rangle \}$ collapses the state to $|0\rangle$ with probability $1 - \epsilon^2$, and to $|1\rangle$ with probability $\epsilon^2$. If $\epsilon = 0$ or $\epsilon = 1$, the post-measurement state coincides with the initial state with certainty—there is no disturbance. More generally, when $\epsilon \approx 0$ or $\epsilon \approx 1$, the post-measurement state remains close to the original one with high probability, indicating low disturbance. In contrast, for $\epsilon = 1/\sqrt{2}$, the state is “maximally far" in the measurement basis, and the outcome is maximally uncertain; the post-measurement state is always far to the initial state, signifying maximal disturbance. This simple example illustrates how disturbance is linked to randomness: measurements that induce minimal disturbance tend to yield more deterministic outcomes, while those that induce maximal disturbance produce outcomes with higher variance.

But how can one perform low-disturbance measurements without prior knowledge of the state? With access to only a single copy, it is fundamentally impossible to design a measurement that avoids disturbing the state while still extracting useful information. However, the situation changes when multiple identical copies are available. In that case, one could strategically use some of them to gain partial information about the state and adapt future measurements to be less disturbing. This naturally leads to the central question of our work:

\begin{center}
\textit{Given access to a finite sequence of an unknown qubit system, what is the best strategy for performing single-copy projective measurements that extract as much information as possible while minimizing the overall disturbance?}
\end{center}

The notion of disturbance is a foundational concept in quantum mechanics and has been explored from various perspectives, notably in works that reformulate the uncertainty principle to quantify the trade-off between measurement-induced disturbance and information gain~\cite{ozawa2003universally, busch2013proof}. Another framework where disturbance is studied  are \emph{weak measurements}, which aim to minimally disturb the quantum system while still providing partial information about it. This idea dates back to the seminal work~\cite{aharonov1988result}, and has since become a central tool in understanding the interplay between information gain and quantum disturbance. However, performing these measurements typically comes at the cost of low information gain: the less disturbing the measurement, the less informative it is about the quantum state, making weak measurements unsuitable for tasks that demand accurate estimation. In contrast, \emph{projective measurements}---which are the focus of this work---provide maximal information about the system but often cause significant disturbance, collapsing the state entirely. 

In our setting, we are not concerned with the disturbance of a single copy, but rather with the \emph{cumulative disturbance} across a sequence of identically prepared quantum states. Our goal is to use each copy as effectively as possible to extract information and achieve sample-optimal estimation of the underlying state. This naturally motivates the design of adaptive measurement strategies that balance the tradeoff between information gain and disturbance over time.

Another related concept is that of \emph{gentle measurements}, which were formalized recently in~\cite{aaronson2019gentle}, but have their roots in earlier work, notably the “gentle measurement” lemma introduced in~\cite{winter2002coding}. These are measurements that guarantee, for certain sets of states, that the post-measurement state remains close to the original one, while still allowing useful information to be extracted. Although this is related to weak measurements, an important distinction is that a gentle measurement is considered weak only if it remains non-disturbing across all states (not only a set). Moreover, this framework does not address how to \emph{adaptively} learn an unknown state using a sequence of projective measurements, which are typically more informative than gentle measurements~\cite{gentle_samplecomplexity}. Our contribution does not lie in proposing a new class of measurements, but rather in developing \emph{adaptive strategies} that employ projective measurements in a way that minimizes cumulative disturbance across the sequence of quantum states.

Formally, we consider a sequential decision-making scenario in which the learner has access to $T$ independent copies of an unknown qubit state $\rho$. At each round $t \in [T]$, the learner selects a measurement direction $|\psi_t\rangle$ and performs a projective measurement in the basis $\{|\psi_t\rangle, |\psi_t^c\rangle\}$. The outcome $r_t \in \{0,1\}$ is sampled according to Born’s rule, $\Pr[r_t = 1] = p_t := \langle \psi_t | \rho | \psi_t \rangle.$ The measurement outcome determines the post-measurement state $\tilde{\psi}_t \in \{|\psi_t\rangle, |\psi_t^c\rangle\}$, with $r_t = 1$ corresponding to the state collapsing to $|\psi_t\rangle$, and $r_t = 0$ to its orthogonal complement.

We quantify the disturbance introduced by the learner through the cumulative expected infidelity between the unknown state $\rho$ and the resulting post-measurement state $\tilde{\psi}_t\in\lbrace | \psi_t \rangle , | \psi^c_t \rangle \rbrace$, compared to the minimal possible disturbance, which occurs when the measurement direction $\psi_t$ is aligned with the eigenvector corresponding to the largest eigenvalue of $\rho$. Formally, we define the cumulative disturbance as
\begin{align}\label{eq:disturbance}
\DIST(T) := \sum_{t=1}^T \left( \mathbb{E}[1 - F(\rho, \tilde{\psi}_t)] - \min_{\tilde{\psi}} \EX [1 - F(\rho, \tilde{\psi})] \right),
\end{align}
where $F(\rho, \sigma) := \left( \Tr \sqrt{ \sqrt{\rho} \sigma \sqrt{\rho} } \right)^2$ denotes the quantum fidelity. Note that the second term in the definition of disturbance, $\min_{\psi}(1 - F(\rho,\psi))$, is constant across rounds and acts as a normalization ensuring that the disturbance vanishes when the learner selects the optimal, least-disturbing measurement direction.  In particular, when $\rho$ is pure, this minimum is zero, as an optimal measurement does not alter the state.
The expression for the disturbance simplifies to the following closed form,
\begin{align}
\DIST(T) = \sum_{t=1}^T 2(\lambda_{\max}(\rho) - p_t)(\lambda_{\max}(\rho) + p_t - 1),
\end{align}
where $\lambda_{\max}(\rho)$ denotes the largest eigenvalue of $\rho$. 

While the above notion of disturbance is defined with respect to the observed outcome, we could also define it as the cumulative infidelity between the unknown state and the average post-measurement state $\rho_t = p_t \psi_t + (1-p_t) \psi_t^c$, i.e 
\begin{align}\label{eq:disturbance2}
  \DIST^* (T) :=& \sum_{t=1}^T  1 - F(\rho,\rho_t).
\end{align}

Although these two notions of disturbance differ in their interpretation—depending on whether the measurement outcomes are observed or not—their behavior is qualitatively the same: both vanish when the measurement direction $\psi_t$ is aligned with the state $\rho$. In particular, when $\rho$ is a pure state, the two disturbances coincide.  Both quantities are controlled by a simpler quantity, the \emph{regret}, defined as
\begin{align}
\REGRET(T) := \sum_{t=1}^T \left( \lambda_{\max}(\rho) - \langle \psi_t | \rho | \psi_t \rangle \right),
\end{align}
since it can be checked that both disturbances satisfy
\begin{align}\label{eq:regret_intro}
\DIST(T) = \Theta(\REGRET(T)) \quad \text{and} \quad \DIST^*(T) = \Theta(\REGRET(T)),
\end{align}
which means that minimizing disturbance is essentially equivalent to minimizing regret. Intuitively, regret remains small when the chosen probe directions are closely aligned with the dominant eigenvector of $\rho$, highlighting that learning the structure of the unknown state is necessary to keep the disturbance low. However, since we are also interested in reconstructing the state, we further require that the learner outputs a final estimate $\hat{\rho}_T$ with high fidelity to the true state after $T$ rounds. That is, in addition to minimizing cumulative disturbance or regret, the algorithm must also achieve low estimation error defined as
\begin{align}\label{eq:fidelity_intro}
    \ERR(T) := 1 - F(\rho,\hat{\rho}_T).
\end{align}

The regret admits a direct physical interpretation in certain quantum thermodynamic scenarios. In particular in~\cite{lumbreras25dissipation} some of the present authors established a connection between regret and the cumulative energy dissipation in quantum state-agnostic work extraction protocols.

Consider a setting in which an unknown source emits identically prepared quantum systems in a fixed state $\rho$. One can design a battery system that sequentially interacts with each quantum copy to extract work from the system and transfer energy into the battery. If the state $\rho$ is known, the protocol can be tailored to extract work optimally at every step. However, when $\rho$ is unknown, each interaction entails a probability of failure due to the mismatch between the protocol and the true state.

This mismatch can be modeled as performing a projective measurement in a guessed direction (corresponding to the control applied to the battery), and success depends on the alignment of this direction with the actual state. In this context, the regret quantifies the cumulative free energy that is wasted due to not applying the optimal work extraction strategy. This interpretation shows the necessity of performing non-invasive tomography \emph{on the fly}, using each quantum copy not only as a source of energy but also as a source of information about the unknown state.

We emphasize that the notion of regret defined in~\eqref{eq:regret_intro} is not merely a formal construction, but a meaningful quantity that captures the cumulative disturbance caused by projective measurements. Moreover, it admits a concrete physical interpretation in quantum thermodynamics, where it corresponds to the total energy dissipation in agnostic work-extraction protocols. Thus, regret serves as both an operational and physical measure of performance in settings that require minimally disturbing projective measurements.

\textbf{Challenges.} We note that the task of minimizing the regret~\eqref{eq:regret_intro} is captured by the \emph{multi-armed quantum bandit (MAQB)} framework introduced in~\cite{lumbreras22bandit} (see also~\cite{brahmachari24intelligence}). This framework was the first to formalize the exploration--exploitation trade-off in online learning of quantum state properties using classical algorithms. In particular, it was shown that when the unknown state $\rho$ is mixed, the regret suffers a fundamental lower bound of order $\REGRET(T) = \Theta(\sqrt{T})$, which is nearly tight, as there exist protocols achieving $\REGRET(T) = \tilde{O}(\sqrt{T})$ by reducing the problem to a linear stochastic bandit~\cite{lattimore_szepesvári_2020} and applying classical bandit algorithms in that setting.

However, this lower bound does not apply when $\rho = | \psi \rangle \! \langle \psi |$ is a pure state. The reason is that the lower bound relies on having vanishing statistical noise in the random outcomes $r_t$, whereas in our setting the shot noise vanishes as the measurement direction $|\psi_t\rangle$ approaches the target state $|\psi\rangle$. In this case, the regret simplifies to
\begin{align}
    \REGRET(T) = \sum_{t=1}^T \left( 1 - F(\psi, \psi_t) \right),
\end{align}
so minimizing regret becomes equivalent to performing \emph{online quantum state tomography with minimal infidelity}, where the goal is to align each probe direction $\psi_t$ as closely as possible to the unknown pure state $\psi$. It is important to emphasize that this notion of regret minimization is not addressed by standard quantum state tomography algorithms, which typically aim to design measurement schemes optimized to output a single classical estimator $\hat{\psi}_T$ minimizing the final estimation error~\eqref{eq:fidelity_intro}, rather than controlling the cumulative error across all measurement rounds. This motivates the following question:

\begin{itemize}
\item \textbf{Question 1.} \textit{Can we perform single copy sample-optimal state tomography in infidelity and achieve at the same time sub-linear regret for unknown pure states? How much adaptiveness is needed for this task?}
\end{itemize}

It is important to note that adaptiveness plays a crucial role for algorithms that aim to minimise the cumulative disturbance of the post-measured state. One could try to use one of the existing sample-optimal algorithms in the incoherent setting such as~\cite{haah2016sample,kueng2017low,guctua2020fast}, which for $T$ samples achieve infidelity $\ERR(T) = O(1/T)$. However, since these algorithms either use fixed bases or randomized measurements, this inevitably leads to a linear scaling $\REGRET(T) = O(T)$. A natural next step is to consider a simple strategy with one round of adaptiveness, where we use a fraction $\alpha \in [0,1] $ of the copies for state tomography to produce a good estimate \(\hat{\psi}\) of the unknown $\psi$, and use the remaining copies to measure along the estimated direction. Using sample-optimal state tomography algorithms this leads to a regret scaling
\begin{align}
    \REGRET(T) = O\left(\alpha T + (T - \alpha T)\frac{1}{\alpha T} \right),
\end{align}
which, optimized over $ \alpha$, gives $\text{Regret}(T) = O(\sqrt{T})$, but results in a sub-optimal error $\ERR(T) = O(1/\sqrt{T})$.

In~\cite{lumbreras22bandit} it was left open the question whether if for pure states one can find an algorithms with a scaling better than $O(\sqrt{T})$ or find a matching lower bound. Since our problem is closely related to the MAQB framework; we name it the pure-state multi-armed quantum bandit (PSMAQB) and use it to address the following questions at the intersection of quantum state tomography and linear stochastic bandits~\cite{lattimore_szepesvári_2020}:

\begin{itemize}
\item \textbf{Question 2.} \textit{Can we break the square root barrier for pure states by showing that $\REGRET(T) = o(\sqrt{T})$ for the PSMAQB problem?}
\end{itemize}
Achieving a better scaling for the PSMAQB problem would provide a physically motivated linear bandit setting where the square root barrier can be surpassed. The linear bandit model with a noise structure studied in~\cite{pmlr-v247-lumbreras24a} is inspired by shot noise in quantum mechanics; however, as we will discuss later, this setting does not align with the PSMAQB problem. The main challenge lies in designing a new algorithm and techniques that exploit the specific structure of the PSMAQB setting compared to the standard linear bandit problem.

\section{Main results} 

In this work, we provide affirmative answers at the same time to Questions 1 and 2 through the following Theorem.

\begin{theorem}[informal]\label{th:abstract_regret_PSMAQB_abstract}
    For any unknown pure qubit state $|\psi\rangle$, we present an algorithm that achieves
    \begin{align}
    \mathbb{E}\left[\textup{Regret}(T) \right] = O \big( \log^2 (T) \big).
    \end{align}
    Moreover, at each time step $t\in [T]$, our algorithm outputs an online estimate $|\hat{\psi}_{t}\rangle$  with infidelity scaling as
    \begin{align}
    \mathbb{E} \left[ 1 -  |\langle \psi | \hat{\psi}_t \rangle|^2\right] = \widetilde{O}\left( \frac{1}{t} \right)
    \end{align}
    Both statements also holds with high probability. 
\end{theorem}

To prove Theorem~\ref{th:abstract_regret_PSMAQB_abstract} we provide an almost fully adaptive adaptive quantum state tomography algorithm that uses $O(T/\log(T))$ rounds of adaptiveness. The exact algorithm and Theorem can be found in Sections~\ref{sec:classical_LinUCBVNN} and~\ref{sec:algorithm_psmaqb}. We say that our algorithm is ``online'' because it is able to output at each time step $t\in [T]$ an estimator with the almost optimal infidelity scaling $O( \frac{1}{t})$ up to logarithmic factors. Now we sketch the main idea of how our algorithm updates the measurements.

\begin{enumerate}
    \item \textbf{Estimation.} At each time step $t\in[T]$  we use the past information of measurements on the direction of $\ket{\psi_{a_1}},...,\ket{\psi_{a_{t-1}}}$ and associated outcomes $r_1,...,r_{t-1}\in\lbrace 0 , 1 \rbrace^{\otimes t-1}$ to build a high probability confidence region $\mathcal{C}_t$ for the unknown environment $|\psi \rangle $.
    \item \textbf{Exploration-exploitation.} A batch of measurements is performed, given by the directions of maximum uncertainty of $\mathcal{C}_t$ such that they give enough information to construct $\mathcal{C}_{t+1}$ (exploration) and also minimise the regret~\eqref{eq:regret_intro} (exploitation).
\end{enumerate}

For the estimation part, we work with the Bloch sphere representation of the unknown state $\Pi = |\psi \rangle \! \langle \psi | = \frac{I+\theta\cdot\sigma}{2}$ where $\theta\in\mathbb{S}^2$ and for $\sigma$ we can take the standard Pauli Basis i.e $\sigma = (\sigma_x,\sigma_y,\sigma_z )$. For each measurement direction $\Pi_{a_t} = |\psi_{a_t}\rangle \! \langle \psi_{a_t} |$, our algorithm performs $k$ independent measurements using the same direction, and it builds the following $k$ online weighted least squares estimators of $\theta$,
\begin{align}\label{eq:abstract_weighted_lse}
    \widetilde{\theta}_{t,i} = V_t^{-1} \sum_{s=1}^t \frac{1}{\hat{\sigma}^2_s (a_s)} r_{s,i} a_s \quad \text{for }i\in[k],
\end{align}
where $r_{s,i}\in\lbrace 0,1\rbrace$ is the outcome of the measurement (up to some renormalization) using the projector $\Pi_{a_s}$ with Bloch vector $a_s\in\mathbb{R}^3$, $V_t =  \mathbb{I} + \sum_{s=1}^t \frac{1}{\hat{\sigma}^2_s (a_s)} a_s a_s^{\mathsf{T}}$ is the design matrix and $\hat{\sigma}^2_s (a_s)$ is a variance estimator of the real variance associated to the outcome $r_s$. The key point where we take advantage from the structure of the quantum problem is that the variance of the outcome $r_a$ associated to the action $\Pi_a$ can be bounded as $\VX [r_a] \leq 1 - |\langle \psi | \psi_a \rangle|^2$. 
The idea is that through a careful choice of actions we can make the terms $1/\hat{\sigma}^2_s (a_s)$ arbitrarily large and ``boost'' the confidence on the directions $a_s$ in the estimators~\eqref{eq:abstract_weighted_lse} that are close to $\theta$. However, this comes at a price, and is that in order to get good concentration bounds for our estimator we need to deal with unbounded random variables and finite variance. We address this issue using the new ideas of median of means (MoM) for online least squares estimators introduced in~\cite{bandits_heavytail,heavy_tail_linear_noptimal,heavy_tail_linear_optimal}. The construction takes inspiration from the old method of median of means~\cite[Chapter 3]{lerasle2019lecture}
 for real random variables with unbounded support and bounded variance but requires non-trivial adaptation for online linear least squares estimators. Similarly to the real case we use the $k$ independent estimators~\eqref{eq:abstract_weighted_lse} in order to construct the MoM estimator $\widetilde{\theta}^{\text{\tiny wMoM}}_{t}$ such that we can build a confidence region with concentration bounds scaling as $1-\exp(-k)$. In particular, the reference we cite for the median of means~\cite{heavy_tail_linear_optimal} is a recent theoretical contribution that explicitly posed as an open question the range of settings where this approach could be applied; our work provides a concrete and significant answer in the quantum domain. We give the exact construction in Section~\ref{sec:MoMLSE}.

\begin{figure}
    \centering
    \begin{overpic}[percent,width=0.5\textwidth]{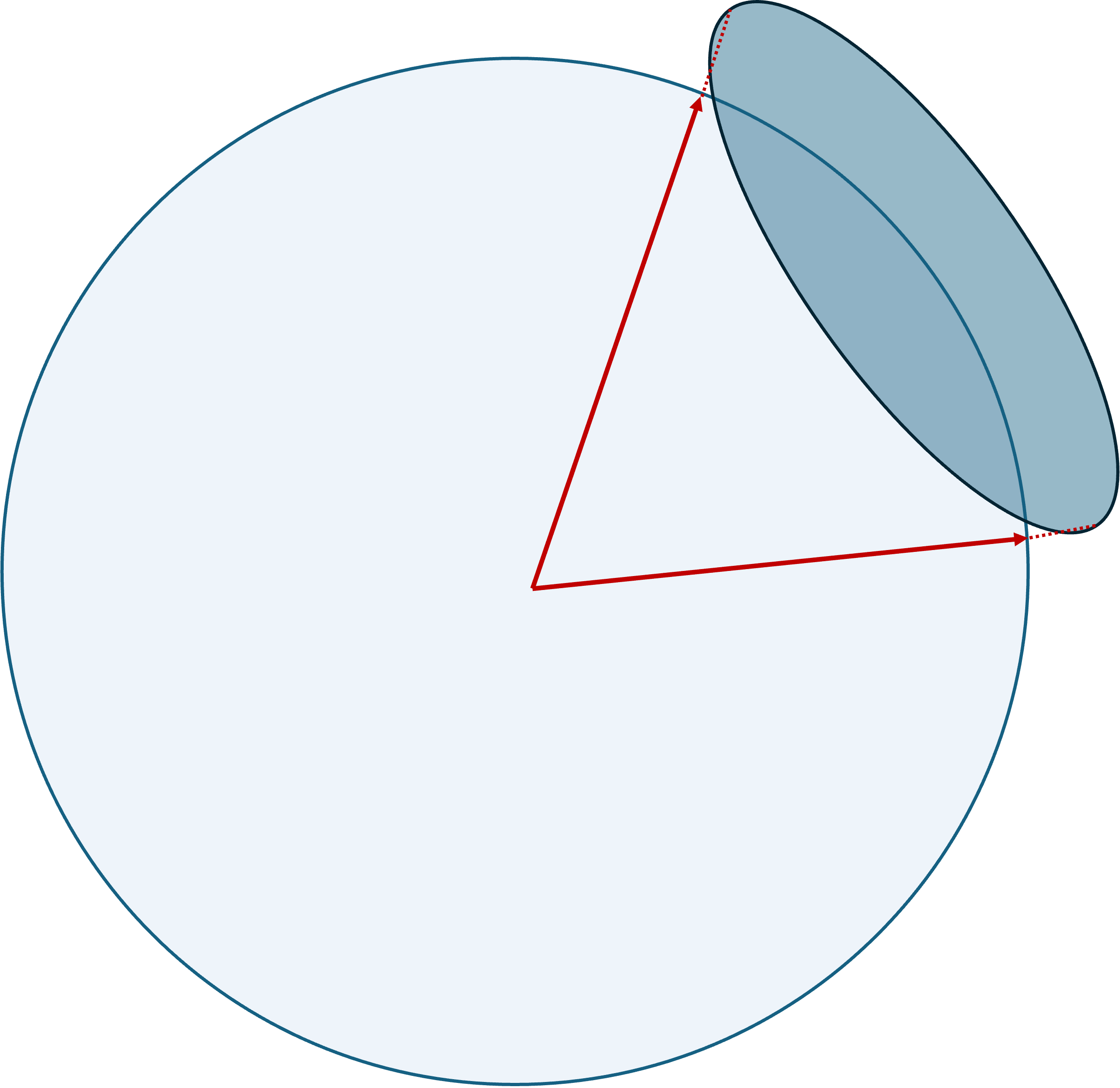}
    \put(50,70){\rotatebox{80}{$|\psi^+_{a_t}\rangle$}}
    \put(70,42){\rotatebox{0}{$|\psi^-_{a_t}\rangle$}}
    \put(75,85){\rotatebox{0}{$\mathcal{C}_t$}}
    \put(82,70){{$|\widehat{\psi}_t\rangle$}}
    \put(90,60){{$| \psi \rangle$}}
    
    \end{overpic}
    \caption{The algorithm at each time step outputs an estimator $|\widehat{\psi}_t\rangle$ and builds a high-probability confidence region $\mathcal{C}_t$ (shaded region) around the unknown state $ | \psi \rangle $ on the Bloch sphere representation. Then uses the optimistic principle to output measurement directions $|\psi^\pm_{a_t} \rangle$ that are close the unknown state $\ket{\psi} $ projecting into the Bloch sphere the extreme points of the largest principal axis of $\mathcal{C}_t$. This particular choice allows optimal learning of $\ket{\psi}$ (exploration) and simultaneously minimizes the regret (exploitation).}
    \label{fig:psmaqb_exploration_exploitation}
\end{figure}

For the exploration-exploitation part, we take the ideas that we develop in~\cite{pmlr-v247-lumbreras24a} (see Figure~\ref{fig:psmaqb_exploration_exploitation}). We give the precise action choice in Section~\ref{sec:linucb_vvn}, and here we sketch the main points. We take inspiration from the optimistic principle for bandit algorithms which in short tells us to choose the most rewarding actions with the available information. In order to use this idea, we use the confidence region that we build in the estimation part and we select measurements that align with the (unknown) direction of $\ket{\psi}$. See Figure~\ref{fig:psmaqb_exploration_exploitation}.
Our algorithm also achieves the relation $1 - |\langle \psi | \psi_{a_t} \rangle |^2 = O\left( 1/\lambda_{\min} (V_t )\right)$, where the minimum eigenvalue $\lambda_{\min}(V_t )$ quantifies the direction of maximum uncertainty (exploration) of our estimator. The maximum eigenvalue $\lambda_{\max}(V_t)$ quantifies the amount of exploitation. We can relate these two concepts through the Theorem we formally state and prove in~\cite[Theorem 3]{pmlr-v247-lumbreras24a}, which states that for our particular measurement choice we have $\lambda_{\min}(V_t) = \Omega (\sqrt{\lambda_{\max}(V_t)} )$. Using this relation and a careful analysis, we can show that $\lambda_{\max}(V_t) = \Omega (t^2)$ which gives $\lambda_{\min}(V_t) = \Omega (t)$ and the scaling $1 -  |\langle \psi | \psi_{a_t} \rangle |^2 = O(1/t)$. We emphasize that the key point that allows to achieve the rate $\lambda_{\min}(V_t) = \Omega (t)$ is the fact that the variance estimators $\hat{\sigma}^2_s$ can get as close as possible to zero since the variance of the rewards goes to zero if we select measurements close to $|\psi \rangle$.

To check the optimality of the regret, we derive a minimax expected regret lower bound based on the optimal quantum state tomography for pure state results in~\cite{hayashi2005reexamination}. The proof does not follow directly from~\cite{hayashi2005reexamination}, and we have to adapt it to the bandit setting. 

\begin{theorem}[informal]
  The cumulative expected regret for any strategy is bounded by
  \begin{align}
    \mathbb{E}\left[ \textup{Regret}(T) \right] = \Omega(\log T),
  \end{align}
  where the expectation is taken over the probability distribution of rewards and actions induced by the learner strategy and also uniformly over the set of pure state environments. 
\end{theorem}

This result is formally derived in Section~\ref{sec:lower_bound}. There it is also generalized to the $d$-dimensional case, in which case the bound is given by $\mathbb{E}\left[ \textup{Regret}(T) \right] = \Omega(d\log (T/d))$. The proof relies on the fact that individual actions of a strategy at time $t\in[T]$ can be viewed as quantum state tomographies using $t$ copies of the state. A relation between the fidelity of these tomographies and the regret of the strategy allows us to convert the fidelity upper bound from~\cite{hayashi2005reexamination} to a regret lower bound. We use measure-theoretic tools to adapt the proof from~\cite{hayashi2005reexamination} to a more general case where the tomography can output an arbitrary distribution of states. We remark that this is a noteworthy result since in~\cite{pmlr-v247-lumbreras24a} they argue how regret lower bound techniques for classical linear bandits fail for noise models with vanishing variance. 

\subsection{Outlook and open problems}

From a quantum state tomography perspective, our work introduces completely new techniques for the adaptive setting, such as the median of means online least squares estimator or the optimistic principle. We expect these techniques to be useful in other quantum learning settings that require adaptiveness, particularly when quantum states serve as resources and must be minimally disturbed during the learning process, such as the state-agnostic work extraction protocols in~\cite{lumbreras25dissipation}. Our algorithm achieves a polylogarithmic regret, which is an exponential
improvement over all previously known algorithms for quantum tomography which can
only achieve such a fidelity by accumulating a linear regret. At a fundamental level, our algorithm goes beyond traditional tomography ideas and shows that is enough to project near the state in order to optimally learn it with minimal disturbance to the samples. From a classical bandit perspective, it is surprising that the setting of learning pure quantum states gives the first non-trivial example of a linear bandit with continuous action sets that achieves polylogarithmic regret. This model motivated our classical work~\cite{pmlr-v247-lumbreras24a} and, jointly with the current work, we establish a bridge between the fields of quantum state tomography and linear stochastic bandits or, more generally, reinforcement learning.

We leave as an open problem the generalization of the algorithm beyond qubits. In particular, our approach relies on the one-to-one correspondence between pure qubit states and the Bloch sphere. While the chosen measurements are specifically designed to work with high-dimensional spheres, for $d > 2$ this correspondence is no longer an isomorphism, and it is not straightforward how to generalize the measurement directions.

We also leave open the question of whether other state tomography algorithms, especially those designed to minimize disturbance, such as weak or gentle measurements, can achieve sublinear regret—particularly polylogarithmic regret. We believe that adaptiveness plays a crucial role in any algorithm aiming to minimize the regret.

\section{The model}\label{sec:model_quantum}

 In this section first connect the notions of disturbance and regret we formally state the PSMAQB problem and make a connection with a linear stochastic bandit problem. Then we define a slightly more general model where the key feature is that the variance of the rewards vanishes with the same behaviour as the PSMAQB problem.

\subsection{Notation}
First, we introduce some basic notation and conventions. Let $[t] = \lbrace 1,2,...,t \rbrace$ for $t\in\mathbb{N}$. For real vectors $x,y\in\mathbb{R}^d$ we denote their inner product as $\langle x, y \rangle = x_1y_1+...+x_dy_d$. Given a real vector $x\in\mathbb{R}^d$ we denote the 2-norm as $\| x \|_2$ and for a real semi-positive definite matrix $A\in\mathbb{R}^{d\times d}$, $A \geq 0$ the weighted norm with $A$ as $\|x \|^2_A = \langle x , A x \rangle$. The set corresponding to the surface of the unit sphere is $\mathbb{S}^{d-1} = \lbrace x\in\mathbb{R}^{d}: \| x \|_2 =1 \rbrace$. For a real symmetric matrix $A\in\mathbb{R}^{d\times d}$ we denote $\lambda_{\max} (A)$, $\lambda_{\min} (A)$ its maximum and minimum eigenvalues respectively. We use the ordering $\lambda_{\min} (A) \leq \lambda_2 (A) , .... , \lambda_{d-1} (A)\leq \lambda_{\max}(A)$ for the $i$-th $\lambda_i (A)$ eigenvalue in increasing order. For a random variable $X$ (discrete or continuous) we denote $\EX [X]$ and $\VX [X] $ its expectation value and variance respectively.  A random variable $X$ is $\sigma$-\textit{subgaussian} if $\forall \lambda\in\mathbb{R},  \EX \left[ \exp(\lambda X)\right]\leq \exp \left(\lambda^2\sigma^2/2\right)$. 

Let $\mathcal{S}_d = \lbrace \rho\in\mathbb{C}^{d\times d}: \Tr(\rho) = 1 , \rho \geq 0\rbrace$ be the set of \textit{quantum states} in a $d$-dimensional Hilbert space $\hil=\mathbb{C}^d$ and $\mathcal{S}^*_d = \lbrace \rho\in\mathcal{S}_d : \rho^2 = \rho \rbrace$ the set of \textit{pure states} or rank-1 projectors.
 We will use the parametrization given in~\cite{byrd2003characterization} of a $d$-dimensional quantum state $\rho_\theta\in\mathcal{S}_d$,
\begin{align}\label{eq:parametrization}
    \rho_\theta = \frac{\mathbb{I}}{d} + \left(\sqrt{\frac{d(d-1)}{2d^2}}\right)\theta\cdot\sigma
\end{align}
where $\theta\in\mathbb{R}^{d^2-1}$, and $\sigma = (\sigma_1,...,\sigma_{d^2-1})$ is a vector of orthogonal, traceless, Hermitian matrices with the normalization condition $\Tr ( \sigma_i \sigma_j ) = 2 \delta_{i,j}$. We will use the subscript $\theta$ in the quantum state $\rho_\theta$ in order to denote the vector of the parametrization~\eqref{eq:parametrization}. In particular the normalization is taken such that $\|\theta\|_2^2\leq 1$ with equality if  $\rho_\theta$ is pure. Note that the parametrization enforces $\rho^\dagger_\theta = \rho_\theta$ and $\Tr(\rho_\theta) = 1$. Also there are some extra conditions on the vector $\theta$ regarding the positivity of the density matrix $\rho_\theta$ but we will not use them. For two quantum states $\rho,\sigma\in\mathcal{S}_d$ the fidelity is $F(\rho,\sigma) = \left(\Tr(\sqrt{\sqrt{\sigma}\rho\sqrt{\sigma}})\right)^2$ and the infidelity $1-F(\rho,\sigma )$. For a Hilbert space $\hil$, the set of linear operators on it will be denoted by $\End(\hil)$. The joint state of a system consisting of $n$ copies of a pure state $\Pi_\theta\in \mathcal{S}_d^*$ is given by the $n$-th tensor power $\Pi_\theta^{\otimes n}\in \End(\hil^{\otimes n})$. Using Dirac notation, we can express $\Pi_\theta=|\psi_\theta\>\!\<\psi_\theta|$ for some normalized $|\psi_\theta\>\in\hil$. Then, the span of all $n$-copy states of the form $|\psi_\theta\>^{\otimes n}$ is called the symmetric subspace of $\hil^{\otimes n}$, denoted by $\hil^{\otimes n}_+$. Its dimension is $D_n=\binom{n+d-1}{d}$. The symmetrization operator $\Pi^+_n\in\End(\hil^{\otimes n})$ is the projector onto $\hil^{\otimes n}_+$. 

\subsection{Cumulative disturbance and regret}

Here we formally show that the notions of disturbance~\eqref{eq:disturbance} and~\eqref{eq:disturbance2} for qubits are indeed controlled by the regret defined as in~\eqref{eq:regret_intro}.

\begin{lemma}
    Consider the notion of disturbance defined in~\eqref{eq:disturbance} then we have that 
    \begin{align}
        \DIST (T) = \Theta (\REGRET (T)),
    \end{align}
    where regret is defined as in~\eqref{eq:regret_intro}.
\end{lemma}

\begin{proof}
    First, without loss of generality we define $ \frac{1}{2} \leq p_t = \langle \psi_t | \rho | \psi_t \rangle$ and using $\psi^c = \mathbb{I} - \psi $ we can directly compute
    \begin{align}
        \mathbb{E}[1 - F(\rho, \tilde{\psi}_t)] = 2 p_t \left( 1 - p_t \right),
    \end{align}
    and using $p_t \leq \lambda_{\max} (\rho) $ we have 
    \begin{align}
        \min_{\psi} \EX[1 - F(\rho, \tilde{\psi})] =  2\lambda_{\max} (\rho) (1 - \lambda_{\max} (\rho)).
    \end{align}
    Then using the identity $1 = x^2 + (1-x)^2 + 2x(1-x)$ we have
    \begin{align}\label{eq:dist_decomp}
        2 p_t \left( 1 - p_t \right) -  2\lambda_{\max} (\rho) (1 - \lambda_{\max} (\rho)) = 2(\lambda_{\max}(\rho) - p_t )(\lambda_{\max}(\rho) + p_t -1).
    \end{align}
    By using $p_t \leq \lambda_{\max} (\rho) \leq 1$ we have 
    \begin{align}
        2(\lambda_{\max}(\rho) - p_t )(\lambda_{\max}(\rho) + p_t -1) \leq  2(\lambda_{\max}(\rho) - p_t ),
    \end{align}
    which leads to $\DIST (T) \leq 2\REGRET (T)$. Then the converse bound follows simply by using $p_t\geq \frac{1}{2}$ in the second factor of~\eqref{eq:dist_decomp}.
\end{proof}

\begin{lemma}
    Consider the notion of disturbance defined in~\eqref{eq:disturbance2}, then we have that 
    \begin{align}
        \DIST^* (T) = \Theta (\REGRET (T)),
    \end{align}
    where regret is defined as in~\eqref{eq:regret_intro}.
\end{lemma}

\begin{proof}
First, without loss of generality we define $ \frac{1}{2} \leq p_t = \langle \psi_t | \rho | \psi_t \rangle$ and then using the clossed formula for the fidelity for qubits and $\rho_t = p_t\psi + (1-p_t)\psi^c_t$, we have
\begin{align}\label{eq:fidelity_postmeasured}
    F(\rho,\rho_t) &= \Tr (\rho\rho_t) + 2\sqrt{\det \rho \det\rho_t} \\
    &= p_t^2 + (1-p_t)^2 + 2\sqrt{\lambda_{\max}(\rho)(1-\lambda_{\max}(\rho))p_t(1-p_t)}.
\end{align}
Using that $p_t \leq \lambda_{\max}(\rho)$ we have
\begin{align}
    F(\rho,\rho_t) &\geq p_t^2 + (1-p_t)^2 + 2p_t(1- \lambda_{\max}(\rho)) \\
    & = 1+2p_t (p_t - \lambda_{\max}(\rho)).
\end{align}
Then we can upper bound the infidelity as
\begin{align}
     1 - F(\rho,\rho_t) \leq (\lambda_{\max}(\rho) - p_t)p_t \leq 2(\lambda_{\max}(\rho) - p_t),
\end{align}
which leads to $\DIST^* (T) \leq 2\REGRET (T)$. For the other bound we can use the geometric mean $2\sqrt{xy}\leq x+y$ in~\eqref{eq:fidelity_postmeasured} and we have
\begin{align}
F(\rho,\rho_t) &\leq p_t^2 + (1-p_t)^2 + p_t (1-p_t) + \lambda_{\max}(\rho)(1- \lambda_{\max}(\rho)) \\
 &= 1 + (p_t-\lambda_{\max}(\rho)) (p_t+\lambda_{\max}(\rho) - 1 ) \leq 1 + (p_t - \lambda_{\max}(\rho)),
\end{align}
where we used $p_t,\lambda_{\max} (\rho)\leq 1$. This gives
\begin{align}
      1 - F(\rho,\rho_t)  \geq \lambda_{\max}(\rho) - p_t ,
\end{align}
which leads to $\DIST^* (T) \geq \REGRET (T)$.
\end{proof}

\subsection{Multi-armed quantum bandit for pure states} 
The model that we are interested in is the general multi-armed quantum bandit model described in~\cite{lumbreras22bandit}[Section 2.3] with the action set being all rank-1 projectors and with pure state environments. For completeness, we state the basic definitions for this particular case for any dimension.
\begin{definition}
 Let $d\in\mathbb{N}$. A  $d$-dimensional \textit{pure state multi-armed quantum bandit} (PSMAQB) is given
by a measurable space $(\mathcal{A}, \Sigma)$, where $\mathcal{A} = \mathcal{S}_d^*$ is the \textit{action set} and $\Sigma$ is a $\sigma$-algebra of subsets of $\mathcal{A}$. The bandit is in an environment, a quantum state $\Pi_\theta\in \mathcal{S}_d^*$, that is unknown.
\end{definition}
The interaction with the PSMAQB is done by a learner that interacts sequentially over $t\in[T]$ rounds with the unknown environment $\Pi_\theta\in\mathcal{S}^*_d$. At each time step $t\in [T]$:
\begin{enumerate}
    \item The learner selects an action $\Pi_{a_t} \in\mathcal{A}$.
    \item Performs a measurement on the unknown environment $\Pi_\theta$ using the two-outcome POVM $\lbrace \Pi_{a_t}, I_{d\times d} - \Pi_{a_t} \rbrace$ and receives a reward $r_t\in\lbrace 0,1 \rbrace$ sampled according to the Born's rule, i.e
    \begin{align}\label{eq:prob_quantum_reward}
    \mathrm{Pr}_{\Pi_\theta} \left(r_t | \Pi_{a_t} \right) = 
    \begin{cases}
    \Tr (\Pi_\theta \Pi_{a_t} ) \quad \text{if} \quad r_t=1,\\
    1-\Tr (\Pi_\theta \Pi_{a_t} ) \quad \text{if} \quad r_t=0, \\
    0 \quad \text{else}.
    \end{cases}
\end{align}
\end{enumerate}
We note that the reward at time step $t$ after selecting $\Pi_{a_t} \in \mathcal{A}$ can be written as
\begin{align}
    r_t =  \Tr (\Pi_\theta \Pi_{a_t} ) + \epsilon_t,
\end{align}
where $\epsilon_t$ is a Bernoulli random variable with values $\epsilon_t \in \lbrace 1- \Tr (\Pi_\theta \Pi_{a_t} ) , - \Tr (\Pi_\theta \Pi_{a_t} )\rbrace$ such that 
\begin{align}\label{eq:variance_psmaqb}
    \EX \left[ \epsilon_t |\mathcal{F}_{t-1} \right] &= 0, \\
    \VX \left[\epsilon_t |\mathcal{F}_{t-1} \right] &= \Tr (\Pi_\theta \Pi_{a_t} ) \left( 1-\Tr (\Pi_\theta \Pi_{a_t} ) \right),
\end{align}
where $\mathcal{F}_{t-1} := \lbrace r_1,\Pi_{a_1},...,r_{t-1},\Pi_{a_{t-1}},\Pi_{a_t} \rbrace $ is a $\sigma$-filtration.

Formally the learner is described by a policy.
\begin{definition}\label{def:policy}
    A policy $\pi$ is a set of conditional probability measures $\lbrace \pi_t \rbrace_{t\in\mathbb{N}}$ on the action set $\mathcal{A}$ of the form
    \begin{align}
        \pi_t ( \cdot |r_1,\Pi_{a_1},...,r_{t-1},\Pi_{a_{t-1}} ) :\Sigma \rightarrow [0,1].
    \end{align}
\end{definition}
Then the policy interacting with the environment $\Pi_\theta$ defines the probability measure over the set of actions and rewards $P_{\Pi_\theta,\Pi}: \left(\Sigma \times \lbrace 0,1\rbrace \right)^{\times T} \rightarrow [0,1]$ as
\begin{align}\label{eq:prob_measure_maqb}
 \int \cdots \int \mathrm{Pr}_{\Pi_\theta} \left(r_T | \Pi_{a_T} \right) \pi_T ( d\Pi_T |r_1,\Pi_{a_1},...,r_{T-1},\Pi_{a_{T-1}} ) \cdots  \mathrm{Pr}_{\Pi_\theta} \left(r_1 | \Pi_{a_1} \right) \pi_1 \left( d\Pi_{a_1} \right),
\end{align}
where the integrals are taken with respect to the corresponding subsets of actions.

The goal of the learner is to efficiently learn a classical description of the environment $\Pi_\theta = \ket{\psi_\theta} \! \bra{\psi_\theta}\in\mathcal{S}_d^*$ while minimizing the disturbance of the post-measured state $\tilde{\psi}_t\in\mathcal{S}^*_d$ that is distributed accordingly to
\begin{align}\label{eq:post_measured_state_dist}
    \mathrm{Pr}\left(R_t | \Pi_{a_t} \right) = \begin{cases}
       \Tr\left( \Pi_\theta \Pi_{a_t} \right) \quad \text{if } \quad R_t = \Pi_{a_t} \\
       1 - \Tr\left( \Pi_\theta \Pi_{a_t} \right) \quad \text{if } \quad R_t = \Pi^c_{a_t} \\
       0 \quad \text{else} ,
    \end{cases} 
\end{align}
where $ \Pi_{a_t} = |\psi_{a_t}\rangle \! \bra{\psi_{a_t}} $ and
\begin{align}
    \Pi^c_{a_t} = |\psi^c_{a_t}\rangle \! \bra{\psi^c_{a_t}} , \quad  |\psi^c_{a_t}\rangle = \frac{\ket{\psi}- \langle  \psi_{a_t}| \psi \rangle \ket{\psi_{a_t}}}{\sqrt{1 - |\langle \psi | \psi_{a_t} \rangle |^2}} .
\end{align}

The task of the learner is captured by minimizing the cumulative regret, our figure of merit that is defined as follows.
\begin{definition}
Given a $d$-dimensional pure state multi-armed quantum bandit, a policy $\pi$, and unknown environment $\Pi_\theta \in \mathcal{S}^*_d$ and $T\in\mathbb{N}$, the $\textit{cumulative regret}$ is defined as
\begin{align}\label{eq:def_regret}
   \textup{Regret}(T,\pi,\Pi_{\theta})   &:= \sum_{t=1}^T 1 - \Tr (\Pi_\theta , \Pi_{a_t}) .
\end{align}
\end{definition}
We note that the regret quantifies the cumulative infidelity between the unknown environment and the post-measured state. And this notion of regret is consistent with the one introduced in the introduction~\eqref{eq:regret_intro} since $\lambda_{\max}(\Pi_\theta) = 1$.

Note that indeed minimizing the regret~\eqref{eq:def_regret} implies selecting actions $\Pi_{a_t}$ that have high fidelity respect to the environment (learning the environment) but at the same time minimizing the cumulative infidelity of the post-measured states. In general the goal of the learner is to minimize the \textit{expected cumulative regret} that is simply defined as $\EX_{\Pi_\theta} [\text{Regret}(T,\pi,\Pi_{\theta}) ]$ where the expectation $\EX_{\Pi_\theta}$ is taken over the probability measure~\eqref{eq:prob_measure_maqb}. When the context is clear, we will use the notation $\text{Regret} ( T )$. Moreover the expression of the regret~\eqref{eq:def_regret} coincides with the notion of regret introduced for general multi-armed bandits~\cite[Section 2.3]{lumbreras22bandit}. For that reason we refer to the PSMAQB problem as the task of finding a policy that minimizes the expected regret $\EX_{\Pi_\theta} [\text{Regret}(T,\pi,\Pi_{\theta}) ]$. Minimizing the regret means achieving sublinear regret on $T$ since $\text{Regret}(T) \leq T$ holds for any policy.

\subsection{Classical model}\label{sec:classical_models}
In order to study the PSMAQB it is helpful to study it using the linear stochastic bandit framework. The idea will be to express the actions and unknown quantum states as real vectors using the parametrization~\eqref{eq:parametrization}.

In the linear stochastic bandit model, the action set is a subset of real vectors i.e $\mathcal{A} \subseteq \mathbb{R}^d$, and the reward at time step $t\in[T]$ after selecting action $a_t\in\mathcal{A}$ is given by
\begin{align}\label{eq:reward_linearbandit}
    r_t = \langle a_t , \theta \rangle + \epsilon_t
\end{align}
where $\theta\in\mathbb{R}^d$ is the unknown parameter and $\epsilon_t$ is some bounded $\sigma-$subgaussian noise that in general can depend on $\theta$ and $a_t$. The regret for this model is given by
\begin{align}\label{eq:linear_regret}
    \text{Regret}_{cl}(T,\pi,\theta) := \sum_{t=1}^T \max_{a\in\mathcal{A}}\langle \theta , a \rangle - \langle \theta , a_t \rangle ,
\end{align}
where the policy $\pi$ is defined analogously to Definition~\ref{def:policy}. We used the subscript $cl$ to differentiate between quantum and classical model.

In order to express the PSMAQB model as a linear stochastic bandit we can use the parametrization~\eqref{eq:parametrization} and express the expected reward for action $\Pi_{a_t}\in\mathcal{S}^*_d$ as
\begin{align}\label{eq:identity_trace}
   \Tr (\Pi_{a_t}\Pi_\theta ) = \frac{1}{d}\left( 1 + \left(d-1 \right) \langle a_t, \theta \rangle \right).
\end{align}
Inverting the above expression we have
\begin{align}\label{eq:inner_product}
    \langle a_t, \theta \rangle = \frac{d\Tr(\Pi_\theta\Pi_{a_t})-1}{d-1}.
\end{align}
Let's quickly revisit the regret expression and use the above identities in order to connect the quantum and classical versions of the regret. We denote $\Pi_{a^*} = \argmax_{\Pi\in\mathcal{A}} \Tr (\Pi \Pi_\theta )$ the optimal action and recall that $1=\Tr (\Pi_{a^*} \Pi_\theta )$. Then we have
\begin{align*}
    \text{Regret}(T,\pi,\Pi_\theta ) &= \sum_{t=1}^T \Tr (\Pi_{a^*} \Pi_\theta ) - \Tr (\Pi_{a_t} \Pi_\theta ) 
    \\&= \frac{d-1}{d} \sum_{t=1}^T \langle \theta, a^* - a_t \rangle.
\end{align*}
Note that by the normalization~\eqref{eq:parametrization} we have that for $\rho_\theta$ and  $\Pi_{a_t}$ the corresponding real vecotrs are normalized $\| \theta \|_2 = \| a_t \| =1$. Thus, since $a^* = \theta$ the regret can be written as
\begin{align}
    \text{Regret}(T,\pi,\Pi_\theta ) &= \frac{d-1}{d} \sum_{t=1}^T \left( 1-\langle \theta, a_t \rangle \right) \\ &= \frac{d-1}{2d} \sum_{t=1}^T \| \theta - a_t \|_2^2 .
\end{align}

Now we want to formulate a classical bandit such that the environment and actions are given by the real vectors that parameterize the quantum states~\eqref{eq:parametrization}. In order to have an expected linear reward that is linear with respect to $\theta$ and $a_t$ it is sufficient to define a renormalized reward as
\begin{align}\label{eq:classical_renormalizedreward}
\tilde{r}_t = \frac{d r_t - 1}{d-1} \in \left\lbrace 1,\frac{-1}{d-1} \right\rbrace,
\end{align}
where we used the reward of the quantum model $r_t\in\lbrace 0,1\rbrace$ given by~\ref{eq:prob_quantum_reward}. Using $\EX[r_t|\mathcal{F}_{t-1}] = \Tr(\Pi_{a_t}\rho_\theta)$ and~\eqref{eq:identity_trace} it is easy to see that
\begin{align}
    \EX[\tilde{r}_t |\mathcal{F}_{t-1}] = \langle \theta , a_t \rangle,
\end{align}
where naturally we use $\mathcal{F}_{t-1} = \lbrace \tilde{r}_1, a_1,...,\tilde{r}_{t-1},a_{t-1},a_t \rbrace.$
Thus, we can write the reward in the form~\eqref{eq:reward_linearbandit}
\begin{align*}\label{eq:variance_linear_classicalquantum}
    \tilde{r_t} = \langle \theta , a_t \rangle + \epsilon_t, \quad \EX [\epsilon_t |\mathcal{F}_{t-1}] = 0 , \\ \VX [\epsilon_t |\mathcal{F}_{t-1}] = \left(1 - \langle \theta , a_t \rangle \right)\left(1 + (d-1)\langle\theta , a_t \rangle\right),
\end{align*}
where the expectation and variance follow from a direct calculation. Then we can study a $d$-dimensional PSMAQB as a linear stochastic bandit choosing the action set
\begin{align}
    \mathcal{A}^{\text{quantum}}_{d} := \lbrace a\in\mathbb{R}^{d^2 - 1} : \Pi_{a}\in\mathcal{S}^*_d \rbrace
\end{align}
with unknown parameter $\theta\in\mathbb{R}^{d^2 - 1}$ such that $\Pi_\theta \in \mathcal{S}^*_d$. The regret of this linear model is given by $\text{Regret}_{cl} = \frac{1}{2}\sum_{t=1}^T \| \theta - a_t \|_2^2$ and we have the following relation with the quantum model:
\begin{align}\label{eq:regret_relation_classicalquantum}
    \text{Regret}(T,\pi,\Pi_\theta) = \frac{d-1}{d}\text{Regret}_{cl}(T,\pi,\theta),
\end{align}
where we take the same strategy $\pi$ in both sides since we can identify the actions of both bandits through the parametrization~\eqref{eq:parametrization} and the relation between rewards given by~\eqref{eq:classical_renormalizedreward}. When the context is clear we will simply use $\text{Regret}(T)$ for both quantum and classical model.

\subsection{Linear bandit with linearly vanishing variance noise}\label{sec:linear_bandit_vanishing_variance}

In~\cite{pmlr-v247-lumbreras24a} some of the present authors introduced the framework of stochastic linear bandits with linear vanishing noise where the setting is a linear bandit with action set $\mathcal{A} = \mathbb{S}^d$, unknown parameter $\theta\in\mathbb{S}^d$ and reward $r_t = \langle \theta , a_t \rangle + \epsilon_t$ such that $\epsilon_t$ is $\sigma_t$-subgaussian with $\EX [\epsilon_t | \mathcal{F}_{t-1}] = 0$ and the property of vanishing noise $\sigma^2_t \leq 1 - \langle \theta , a_t \rangle^2$. In order to study a PSMAQB we will relax the condition on the subgaussian noise and we will replace it by the following condition on the noise
\begin{align}\label{eq:noise_condition_variance_vanishing}
  \EX \left[ \epsilon_t | \mathcal{F}_{t-1 } \right] = 0, \quad \VX \left[ \epsilon_t | \mathcal{F}_{t-1} \right] \leq 1 - \langle \theta , a_t \rangle^2 .
\end{align}
As in the classical model of the previous section using that $\max_{a\in\mathcal{A}} \langle \theta , a \rangle =1$ we have that the regret is given by
\begin{align}\label{eq:regret_variance_model}
    \text{Regret}(T) = \sum_{t=1}^T 1-\langle \theta , a_t \rangle = \frac{1}{2}\sum_{t=1}^T \| \theta - a_t \|_2^2.
\end{align}
We note that finding a strategy that minimizes regret for the above model will also work for a $d=2$ PSMAQB with unknown $\Pi_{\theta}\in\mathcal{S}^*_2$ using the relations of last sections since
\begin{align}
    \mathcal{A}_2^{\text{quantum}} = \lbrace a\in\mathbb{R}^3: \|a \|_2 = 1 \rbrace = \mathbb{S}^2 ,
\end{align}
and the variance of the PSMAQB~\eqref{eq:variance_linear_classicalquantum} fullfills the relation~\eqref{eq:noise_condition_variance_vanishing}.

\section{Algorithm for bandits with linearly vanishing variance noise}\label{sec:classical_LinUCBVNN}
In this Section we are going to present an algorithm for the linear bandit model explained in Section~\ref{sec:linear_bandit_vanishing_variance} that is based on the algorithm \textsf{LINUCB-VN} studied in~\cite{pmlr-v247-lumbreras24a} for linear bandits with linearly vanishing noise. Later we will show how to use this algorithm for the qubit PSMAQB problem. 

\subsection{Median of means for an online least squares estimator}\label{sec:MoMLSE}
First we discuss the medians of means method for the online linear least squares estimator introduced in~\cite{heavy_tail_linear_optimal}. We are going to use this estimator later in order to design a strategy that minimizes the regret for the model introduced in Section~\ref{sec:linear_bandit_vanishing_variance}. The reason we need this estimator is that in the analysis of our algorithm we need concentration bounds for linear least squares estimators where the random variables have bounded variance and a possibly unbounded subgaussian parameter. The condition of bounded variance is weaker than the usual assumption of bounded subgaussian noise, however we can recover similar concentration bounds of the estimator if we implement a median of means.

In order to build the median of means online least squares estimator for linear bandits we need to sample $k$ independent rewards for each action. Specifically given an action set $\mathcal{A}\subset\mathbb{R}^d$, an unknown parameter $\theta\in\mathbb{R}^d$, at each time step $t$ we select an action $a_t\in\mathcal{A}$ and sample $k$ independent rewards using $a_t$ where the outcome rewards are distributed as
\begin{align}
    r_{t,i} = \langle \theta , a_t \rangle + \epsilon_{t,i} \quad \text{for }i\in[k],
\end{align}
for some noise such that $\EX [\epsilon_{t,i} | \mathcal{F}_{t-1} ] = 0$. We refer to $k$ as the number of subsamples per time step. Then at time step $t$ we define $k$ least squares estimators as
\begin{align}
    \widetilde{\theta}_{t,i} = V_t^{-1} \sum_{s=1}^t r_{s,i} a_s \quad \text{for }i\in[k],
\end{align}
where $V_t$ is the design matrix defined as
\begin{align}
    V_t = \lambda \mathbb{I} + \sum_{s=1}^t a_s a_s^{\mathsf{T}},
\end{align}
with $\lambda > 0$ being a parameter that ensures invertibility of $V_t$. We note that the design matrix is independent of $i$. Then the median of means for least squares estimator (MOMLSE) is defined as
    \begin{align}\label{eq:linearMOM}
        \widetilde{\theta}_{t}^{\text{\tiny MoM}} := \tilde{\theta}_{t,k^*} \quad \text{where }k^* = \argmin_{j\in[k]} y_j ,
    \end{align}
where     
\begin{align}
        y_j = \text{median}\lbrace \|\tilde{\theta}_{t,j} - \tilde{\theta}_{t,i} \|_{V_t}: i\in [k]/j \rbrace \quad \text{for } j \in [k].
    \end{align}

Using the results in~\cite{heavy_tail_linear_optimal} we have that the above estimator has the following concentration property around the true estimator.
\begin{lemma}[Lemma 2 and 3 in~\cite{heavy_tail_linear_optimal}]\label{lem:concentration_mom}
Let $\widetilde{\theta}_t^{\textup{\tiny MoM}}$ be the MOMLSE defined~\eqref{eq:linearMOM} in with $k$ subsamples with  $\lbrace r_{s,i}\rbrace_{(s,i)\in [t]\times [k]} $ rewards and corresponding actions $\lbrace a_{s} \rbrace_{s\in [t]} $. Assume that the noise of all rewards has bounded variance, i.e $\EX \left[ \epsilon^2_{s,i} | \mathcal{F}_{t-1} \right] \leq 1$ for all $s\in[t]$ and $i\in [ k ]$. Then we have
    \begin{align}
    \mathrm{Pr}\left( \|\theta - \widetilde{\theta}_t^{\text{\tiny MoM}} \|^2_{V_t} \leq 9\left(\sqrt{9d} + \lambda \| \theta \|_2 \right)^2 \right)\geq 1 - \exp \left( \frac{-k}{24}\right).
\end{align}

\end{lemma}

We will use a slight modification of the above result with a weighted least squares estimator like the one used in~\cite{pmlr-v247-lumbreras24a}. The weights will be related to a variance estimator of the noise for action $a\in\mathcal{A}$ that at each time step $t$ can be generally defined as
\begin{align}\label{eq:variance_estimator}
    \hat{\sigma}^2_t : \mathcal{H}_{t-1}\times A \rightarrow \mathbb{R}_{>0},
\end{align}
where $\mathcal{H}_{t-1} = \lbrace r_{s,i}\rbrace_{(s,i)\in [t-1]\times [k]} \cup \lbrace a_{s} \rbrace_{s \in [t-1]}$ contains the past information of rewards and actions played. For our purposes we will use only the information of the past actions and in order to simplify notation we will use $\hat{\sigma}^2_t (a)$ to denote an estimator of the variance for the reward associated action $a\in\mathcal{A}$ with the information collected up to time step $t-1$. Then the corresponding weighted versions with $k$ subsamples are defined as
\begin{align}\label{eq:weighted_lse}
    \widetilde{\theta}_{t,i} = V_t^{-1} \sum_{s=1}^t \frac{1}{\hat{\sigma}^2_s (a_s)} r_{s,i} a_s \quad \text{for }i\in[k],
\end{align}
with the weighted design matrix
\begin{align}\label{eq:weighted_design}
    V_t = \lambda \mathbb{I} + \sum_{s=1}^t \frac{1}{\hat{\sigma}^2_s (a_s)} a_s a_s^{\mathsf{T}}.
\end{align}
Then the weighted version of the median of means linear estimator is defined analogously to~\eqref{eq:linearMOM} with the corresponding weighted versions~\eqref{eq:weighted_lse}\eqref{eq:weighted_design} and we will denote it as $\widetilde{\theta}_t^{\text{\tiny wMOM}}$.
In our algorithm analysis we will use the following analogous concentration bound under the condition that the estimators $\hat{\sigma}^2_t$ overestimate the true variance.
\begin{corollary}\label{cor:concentration_wmom}
    Let $\widetilde{\theta}_t^{\textup{\tiny wMOM}}$ be the weighted version of the MOMLSE with $k$ subsamples, $\lbrace r_{s,i}\rbrace_{(s,i)\in [t]\times [k]} $ rewards with corresponding actions $\lbrace a_{s} \rbrace_{s\in [t]} $ and variance estimator $\hat{\sigma}^2_t$. Define the following event
    \begin{align}\label{eq:gt_event}
        G_t : = \lbrace \big( \mathcal{H}_{t-1}, a_t \big) : \VX [\epsilon_{s,i}] \leq \hat{\sigma}^2 (a_s )\ \forall s,i\in [t]\times [k  ]\rbrace.
    \end{align}
    Then we have
    \begin{align}
    \mathrm{Pr}\left( \|\theta - \widetilde{\theta}_t^{\textup{\tiny wMOM}} \|^2_{V_t} \leq \beta \mid G_t \right)\geq 1 - \exp \left( \frac{-k}{24}\right),
    \end{align}
    where 
    \begin{align}\label{eq:beta_constant}
        \beta := 9\left(\sqrt{9d} + \lambda \| \theta \|_2 \right)^2 .
    \end{align}
\end{corollary}

\begin{proof}
    The result follows from applying Lemma~\ref{lem:concentration_mom} to the sequences of re-normalized rewards $\lbrace \frac{r_{s,i}}{\hat{\sigma}_s (a_s)}\rbrace_{(s,i)\in [t]\times [k]} $  and actions $\lbrace \frac{a_{s,i}}{\hat{\sigma}_s (a_s)} \rbrace_{s\in [t]} $. We only need to check that the sequence $\lbrace \frac{\epsilon_{s,i}}{\hat{\sigma}_s (a_s)}\rbrace_{(s,i)\in [t]\times [k]} $ has finite variance. Conditioning with the event $G_t$ and the fact that by definition $\hat{\sigma}^2_s (a_s)$ only depend on the past $s-1$ action and rewards we have that the re-normalized noise has bounded variance since
    \begin{align}
        \EX \left[ \left(\frac{\epsilon_{s,i}}{\hat{\sigma}_s (a_s)} \right)^2 \Bigg| \mathcal{F}_{t-1 }\right] = \frac{1}{\hat{\sigma}^2_s (a_s)}\EX [ \epsilon^2_{s,i} | \mathcal{F}_{t-1} ] = \frac{\VX [\epsilon_{s,i}]}{\hat{\sigma}^2_s (a_s)}\leq 1.
    \end{align}
\end{proof}

\subsection{Algorithm}\label{sec:linucb_vvn}
The algorithm that we design for linear bandits with linearly variance vanishing noise is \textsf{LinUCB-VVN} (\textsf{LinUCB} vanishing variance noise) stated in Algorithm~\ref{alg:linucb_vn_var}. The algorithm updates the actions in batches of lenght $2k(d-1)$. For every batch it outputs $2(d-1)$ actions and samples $k$ independent reward with each action. We use a slightly abuse of notation and label each batch with $t$. For each batch $t\geq 1$ the actions are updated as

\begin{align}\label{eq:action_general_update}
    a^\pm_{t,i} := \frac{\widetilde{a}^\pm_{t,i}}{\| \widetilde{a}^\pm_{t,i}\|_2}, \quad  \tilde{a}^\pm_{t,i} = \frac{\widetilde{\theta}^{\text{wMoM}}_{t-1}}{\| \widetilde{\theta}^{\text{wMoM}}_{t-1} \|_2} \pm \frac{1}{\sqrt{\lambda_{\min}(V_{t-1})}} v_{t-1,i} ,
\end{align}
for $i\in [d-1 ]$, $v_{t-1.i}$ is the normalized eigenvector of $V_{t-1}$ with eigenvalue $\lambda_i (V_{t-1} )$ and $\widetilde{\theta}^{\text{wMoM}}_t$ is the weighted MOMLSE defined as in Section~\ref{sec:MoMLSE} that is build with the $k$ sampled rewards of each action. The design matrix $V_t$ is updated as
\begin{align}\label{eq:design_update}
    V_t = V_{t-1} + \omega (V_{t-1} ) \sum_{i=1}^{d-1} \left( a^+_{t,i} (a^+_{t,i})^{\mathsf{T}} + a^-_{t,i} (a^-_{t,i})^\mathsf{T} \right)
\end{align}
where the weights $\omega$ and variance estimator are chosen as
\begin{align}\label{eq:weight_choice}
    \omega (V_{t-1}) := \frac{\sqrt{\lambda_{\max}(V_{t-1})}}{12\sqrt{d-1}\beta} , \quad \hat{\sigma}^2_t (a^\pm_{t,i}) := \frac{1}{\omega (V_{t-1})}.
\end{align}
We note that the definition for $\hat{\sigma}^2_t (a^\pm_{t,i})$ fulfills the definition of variance estimator~\eqref{eq:variance_estimator} stated in the previous section since it only depends on the past history $\mathcal{H}_{t-1}$.
\begin{algorithm}
	\caption{\textsf{LinUCB-VVN}} 
	\label{alg:linucb_vn_var}
 
        Require: $\lambda_0\in\mathbb{R}_{>0}$, $k\in\mathbb{N}$,  $\omega: \text{P}^d_+ \rightarrow \mathbb{R}_{\geq 0}$
        
        Set initial design matrix $V_0 \gets \lambda_0\mathbb{I}_{d\times d}$ 
        
        Choose initial estimator ${\theta}_0\in\mathbb{S}^d$ for $\theta$ at random 
        
        \For{$t=1,2,\cdots$}{
            \vspace{1mm}
            \textit{Optimistic action selection}
            \vspace{1mm}
            
            \For{$i = 1,2,\cdots d-1$}{    
                Select actions $a^+_{t,i}$ and $a^-_{t,i}$ according to Eq.~\eqref{eq:action_general_update}
                
                \vspace{1mm}
                \textit{Sample $k$ independent rewards for each $a^\pm_{t,i}$}
                \vspace{1mm}
                
                \For{$j=1,...,k$}{
                    Receive associated rewards $r^+_{t,i,j}$ and $r^-_{t,i,j}$
                }
            }
            
            \vspace{1mm}
            \textit{Update variance estimator for $a^+_{t,i}$}
            \vspace{1mm}
            
            $\hat{\sigma}^2_t \gets \frac{1}{\omega ( V_{t-1}(\lambda_0 ) ) }$ for $t\geq 2$ or $\hat{\sigma}^2_t \gets 1$ for $t=1$.

            \vspace{1mm}
            \textit{Update design matrix}
            \vspace{1mm}
            
            $V_{t} \gets V_{t-1} + \frac{1}{\hat{\sigma}_t^2} \sum_{i=1}^{d-1} \left(a^+_{t,i}  (a^+_{t,i})^{\mathsf{T}}  +  a^-_{t,i}  (a^-_{t,i})^\mathsf{T} \right)$

            \vspace{1mm}
            \textit{Update LSE for each subsample}
            \vspace{1mm}

            \For{$j=1,2,...,k$:}{
                $\widetilde{\theta}_{t,j}^\text{w} \gets V_t^{-1}  \sum_{s = 1}^t \frac{1}{\hat{\sigma}_t^2} \sum_{i=1}^{d-1} (a^+_{s,i} r^+_{s,i,j} + a^-_{s,i} r^-_{s,i,j} ) $
            }   
            Compute $\widetilde{\theta}_t^{\text{\tiny wMOM}}$ using $\lbrace \widetilde{\theta}_{t,j}^\text{w} \rbrace_{j=1}^k$
        }
\end{algorithm}

\subsection{Regret analysis}
In this Section we present the analysis of the regret for Algorithm~\ref{alg:linucb_vn_var}. The analysis is similar to the \textsf{LinUCB-VN} presented in~\cite{pmlr-v247-lumbreras24a}[Appendix C.1]. Thus, we focus on the changes respect to \textsf{LinUCB-VN} and although we present a complete proof we refer to~\cite{pmlr-v247-lumbreras24a} for more detailed computations. The main result we use from~\cite{pmlr-v247-lumbreras24a} is a theorem that quantifies the growth of the maximum and minimum eigenvalues of the design matrix $V_t$~\eqref{eq:design_update}.

\begin{theorem}[Theorem 3 in~\cite{pmlr-v247-lumbreras24a}]\label{th:eigenvalue_scaling}
    Let $\lbrace c_t \rbrace_{t=0}^\infty \subset \mathbb{S}^{d-1}$ be a sequence of normalized vectors and $\omega: \textup{P}^d_+ \rightarrow \mathbb{R}_{\geq 0}$ a function such that 
    \begin{align}
        \omega (X) \leq C\sqrt{\| X \|_\infty},
    \end{align}
     for a constant $C > 0$ and any $X\in \textup{P}^d_+$. Let $\lambda_0 \geq \max \big\lbrace 2,\sqrt{\frac{2}{3(d-1)}}2dC+\frac{2}{3(d-1)} \big\rbrace$, and define a sequence of matrices
     $\lbrace V_t \rbrace_{t=0}^\infty \subset \mathbb{R}^{d\times d}$ as
       \begin{align}\label{eq:vt_lemma}
         V_0 := \lambda_0\mathbb{I}_{d\times d}, \quad      V_{t+1} := V_t + \omega ( V_t ) \sum_{i=1}^{d-1}P_{t,i}, 
       \end{align}
       where 
       \begin{align}\label{eq:defP_a}
           P_{t,i} : = a^+_{t+1,i}(a^+_{t+1,i})^\mathsf{T}  +  a^-_{t+1,i} (a^-_{t+1,i})^\mathsf{T} ,\\
           a^\pm_{t+1,i} : = \frac{\tilde{a}^\pm_{t+1,i}}{\| \tilde{a}^\pm_{t+1,i}\|_2}, \quad \tilde{a}^\pm_{t+1,i} := c_t \pm \frac{1}{\sqrt{\lambda_{t,1}}} v_{t,i},
       \end{align}
       with $\lambda_{t,i} = \lambda_{i}(V_t)$ the eigenvalues of $V_t$ with corresponding normalized eigenvectors $v_{t,1},...,v_{t,d}\in\mathbb{S}^{d-1}$.
  Then we have
    \begin{align}\label{eq:eig_relation}
        \lambda_{\min}(V_t) \geq \sqrt{\frac{2}{3(d-1)}\lambda_{\max}(V_t)} \quad \text{for all}\quad t\geq 0.
    \end{align}
  
\end{theorem}
For the proof of the above Theorem we refer to the original reference. Then using this Theorem and the concentration bound for MOMLSE given in Corollary~\ref{cor:concentration_wmom} we can provide the following regret analysis for a stochastic linear bandit with vanishing variance noise.

\begin{theorem}\label{th:regret_scaling_variance}
Let $d\geq 2$, $k\in\mathbb{N}$ and $T=2(d-1)k\widetilde{T}$ for some $\widetilde{T}\in\mathbb{N}$, $\widetilde{T}\geq 2$. Let $\omega (X)$ defined as in~\eqref{eq:weight_choice} using $\lambda_0$ satisfying the constraints in Theorem~\ref{th:eigenvalue_scaling}. Then if we apply \textsf{LinUCB-VVN}~\ref{alg:linucb_vn_var}($\lambda_0,k,\omega$) to a $d$ dimensional stochastic linear bandit with variance as in~\eqref{eq:noise_condition_variance_vanishing} with probability at least $(1-\exp (-k/24))^{\widetilde{T}}$ the regret satisfies
      \begin{align}
      \textup{Regret}(T) \leq 4k(d-1)+144d(d-1)k\beta^2\log\left(\frac{T}{2(d-1)k} \right) \\ +24(d-1)^{\frac{3}{2}}k\beta\log\left(\frac{T}{2(d-1)k} \right),
    \end{align}
    and at each time step $t\in[T]$ with the same probability it can output an estimator $\hat{\theta}_t\in\mathbb{S}^{d-1}$ such that
    \begin{align}
        \|\theta - \hat{\theta}_t\|_2^2 \leq \frac{576d^2\beta^2 k+96d\sqrt{d-1}\beta k}{t},
    \end{align}
    with $\beta$ defined as in~\eqref{eq:beta_constant}.
\end{theorem}

From the above Theorem we have that if we set $k = \lceil 24\log\left( \frac{\widetilde{T}}{\delta}\right) \rceil$ for some $\delta \in \left( 0 , 1 \right)$ then with probability at least $1 - \delta$ \textsf{LinUCB-VNN} achieves
\begin{align}
    \text{Regret}(T) = O\left( d^4\log^2 (T)\right), \quad  \|\theta - \hat{\theta}_t\|_2^2 = O\left( \frac{\log(T)}{t}\right).
\end{align}

\begin{proof}
From the expression of the regret~\eqref{eq:regret_variance_model} we have that to give an upper bound it suffices to gives an upper bound between the distance of the unknown parameter $\theta$ and the actions $a^\pm_{t,i}$ selected by the algorithm~\eqref{eq:action_general_update}. We denote the step $\tilde{t}\in [\widetilde{T} ]$ to run over the batches the algorithm updates the MoM estimator $\widetilde{\theta}_t^{\text{\tiny wMOM}}$. First we will do the computation assuming that the event
\begin{align}
    E_{\tilde{t}} := \lbrace \mathcal{H}_{\tilde{t}} : \forall s \in [\tilde{t}], \theta \in \mathcal{C}_s \rbrace , 
\end{align}
holds where $\mathcal{C}_s = \lbrace \theta' \in\mathbb{R}^d : \| \theta' - \widetilde{\theta}^{wMOM}_{\tilde{t}} \|^2_{V_s} \leq \beta \rbrace$. Here the history $\mathcal{H}_{\tilde{t}}$ is defined with the previous outcomes and actions of our algorithm i.e 
\begin{align}
    \mathcal{H}_{\tilde{t}} := \left(r^+_{s,i,j}, a^+_{s,i}, r^-_{s,i,j} , a^-_{s,i} \right)_{(s,i,j)\in[\tilde{t}]\times[d-1]\times [k]} 
\end{align}
Later we will quantify the probability that this event always hold. Using the definition of the actions~\eqref{eq:action_general_update}, $\theta,\widetilde{\theta}^{\text{\tiny wMOM}}_{\tilde{t}}\in\mathbb{S}^{d-1}$ and the arguments from~\cite{pmlr-v247-lumbreras24a}[Appendix C.1, Eq. (165)] we have that 
\begin{align}
    \|\theta - a^{\pm}_{\tilde{t},i} \|_2^2 \leq \frac{9\beta}{\lambda_{\min}(V_{\tilde{t}-1})}.
\end{align}
Then using that the design matrix $V_{\tilde{t}}$~\eqref{eq:design_update} is updated as in Theorem~\ref{th:eigenvalue_scaling} and the choice of weights~\eqref{eq:weight_choice} we fix 
\begin{align}\label{eq:lambda_condition}
    \lambda_0 \geq \max \left\lbrace 2, 2\sqrt{\frac{2}{3(d-1)}}\frac{d}{12\sqrt{d-1}\beta}+\frac{2}{3(d-1)} \right\rbrace
\end{align}
and we have that $\lambda_{\min}(V_{\tilde{t}})\geq \sqrt{\frac{2}{3(d-1)}\lambda_{\max}(V_{\tilde{t}})}$ applying Theorem~\ref{th:eigenvalue_scaling}. Inserting this into the above we have
\begin{align}\label{eq:dist_attheta}
   \|\theta - a^{\pm}_{\tilde{t},i} \|_2^2 \leq \frac{12\sqrt{d-1}\beta}{\sqrt{\lambda_{\max}(V_{\tilde{t}})}}.
\end{align}
Thus, it remains to provide a lower bound on $\lambda_{\max}(V_{\tilde{t}})$. We note that in~\cite{pmlr-v247-lumbreras24a}[Appendix C.1] they also had to provide an upper bound but this was because the constant $\beta$ beta they use depends on $t$. From the definition of $V_t$~\eqref{eq:design_update} we can bound the trace as
\begin{align}
    \mathrm{Tr}(V_{\tilde{t}}) &\geq \sum_{s=2}^{\tilde{t}} 2(d-1)\omega (V_{s-1} )\\
    &= \frac{\sqrt{d-1}}{6\beta}\sum_{s=1}^{\tilde{t}-1} \sqrt{\lambda_{\max}(V_s)}.
\end{align}
Then using the bound $\mathrm{Tr}(V_{\tilde{t}}) \geq \lambda_{\max}(V_{\tilde{t}})/d$ and some algebra we arrive at
\begin{align}
    \lambda_{\max}(V_{\tilde{t}}) \geq \frac{1}{1+6\frac{d}{\sqrt{d-1}}\beta}\sum_{s=1}^{\tilde{t}} \sqrt{\lambda_{\max}(V_s)}.
\end{align}
Now we have an inequality with the function $\lambda_{\max}(V_s)$ at both sides. In order to solve it we use the technique from~\cite{pmlr-v247-lumbreras24a}[Appendix C.1, Eqs.~(197)--(208)] which consist on extending $\lambda_{\max}(V_{\tilde{t}})$ to the continuous with a linear interpolation and then transforming the sum to an integral which leads to a differential inequality. Solving this leads to 
\begin{align}
    \lambda_{\max}(V_{\tilde{t}}) \geq \frac{\tilde{t}^2}{4(1+6\frac{d}{\sqrt{d-1}}\beta)^2}.
\end{align}
Now we can insert the above into~\eqref{eq:dist_attheta} and we have
\begin{align}\label{eq:theta_at_distance}
     \|\theta - a^{\pm}_{\tilde{t},i} \|_2^2 &\leq \frac{24\sqrt{d-1}\beta(1+6\frac{d}{\sqrt{d-1}}\beta)}{\tilde{t}-1} \\
     &= \frac{144d\beta^2+24\sqrt{d-1}\beta}{\tilde{t}-1}.
\end{align}
Thus, we can inserted the above bound into the regret expression~\eqref{eq:regret_variance_model} and we have
\begin{align}
    &\text{Regret}(T) = \frac{1}{2}\sum_{t=1}^{T} \|\theta - a_t \|_2^2   \\
    &= \frac{1}{2}\sum_{\tilde{t}=1}^{\tilde{T}}\sum_{i=1}^{d-1}\sum_{j=1}^k \left( \|\theta - a^{+}_{\tilde{t},i} \|_2^2 + \|\theta - a^{-}_{\tilde{t},i} \|_2^2 \right) \\
    &\leq  4k(d-1) + \frac{1}{2}\sum_{\tilde{t}=2}^{\tilde{T}}\sum_{i=1}^{d-1}\sum_{j=1}^k \left( \|\theta - a^{+}_{\tilde{t},i} \|_2^2 + \|\theta - a^{-}_{\tilde{t},i} \|_2^2 \right)  \\
    &\leq 4k(d-1) + (144d(d-1)k\beta^2+24(d-1)^{\frac{3}{2}}k\beta) \sum_{\tilde{t}=2}^{\tilde{T}} \frac{1}{t-1}\\
    &\leq 4k(d-1)+144d(d-1)k\beta^2\log \widetilde{T} +24(d-1)^{\frac{3}{2}}k\beta\log\widetilde{T}  \\
    &=  4k(d-1)+144d(d-1)k\beta^2\log\left(\frac{T}{2(d-1)k} \right) \nonumber \\ &\quad +24(d-1)^{\frac{3}{2}}k\beta\log\left(\frac{T}{2(d-1)k} \right).
\end{align}
It remains to quantify the probability that the event $E_{\tilde{t}}$ holds. For that we will use the concentration bounds of the median of means for least squares estimator stated in Corollary~\ref{cor:concentration_wmom}. From the variance condition of our model~\eqref{eq:noise_condition_variance_vanishing} we have
\begin{align}
    \VX [\epsilon^\pm_{\tilde{t},i,j} | \mathcal{F}_{\tilde{t}-1}] \leq 1 - \langle \theta , a^\pm_{\tilde{t},i} \rangle^2 \leq 2(1-\langle \theta , a^\pm_{\tilde{t},i} )) = \| \theta - a^\pm_{\tilde{t},i} \|_2^2,
\end{align}
where we used $1+\langle \theta , a^\pm_{\tilde{t},i} \rangle \leq 2$.
Thus from our choice of weights~\eqref{eq:weight_choice} and~\eqref{eq:theta_at_distance} we have that
\begin{align}
    \text{if } \theta\in\mathcal{C}_{s-1} \Rightarrow \VX [\epsilon^\pm_{\tilde{t},i,j} | \mathcal{F}_{\tilde{t}-1}] \leq \hat{\sigma}_s^2 (a^\pm_{s,i}).
\end{align}
Then in order to apply Corollary~\ref{cor:concentration_wmom} we note that from the choice $\hat{\sigma}_s^2 (a^\pm_{1,i}) = 1$ the event $G_{\tilde{t}}$ at $\tilde{t}=1$ is always satisfied i.e $\mathrm{Pr}(G_1) = 1$. Then applying Bayes theorem, union bound over the events $G_1,E_1,...,G_{t-1},E_t$ and Corollary~\ref{cor:concentration_wmom} we have
\begin{align}
    \mathrm{Pr}(E_{\widetilde{T}} \cap G_{\widetilde{T}}) \geq \left(1 - \exp (-k/24) \right)^{\widetilde{T}} .
\end{align}   
This probability also quantifies the probability that~\eqref{eq:theta_at_distance} holds since the only assumption we used is $\theta\in\mathcal{C}_{\tilde{t}-1}$. Then we can take simply one of the actions $a^\pm_{\tilde{t},i}$ as the estimator $\hat{\theta}_t$ and the result follows using the relabeling $t = 2(d-1)k\tilde{t}$ and the inequality $1/(\tilde{t} - 1) \leq 2/\tilde{t}$ for $\tilde{t}\geq 2$.  A more detailed analogous computation of the above probability can be found in~\cite{pmlr-v247-lumbreras24a}[Appendix C.1].
\end{proof}

In the previous Theorem we did not set a specific value for the parameter $k$ or the number of subsamples per action. We note that the regret scales linearly with $k$ but since the success probability scales exponentially with $k$ it will suffice to set $k\sim \log (T)$ such that in expectation we get the $\log^2 (T)$ behavior. We formalize this in the following Corollary.

\begin{corollary}\label{cor:expected_regret}
Under the same assumptions of Theorem~\ref{th:regret_scaling_variance} we can fix $k = \lceil 24\log(\widetilde{T}^2) \rceil$ and we have that for $t\in [T ]$,
\begin{align}
    \EX \left[ \textup{Regret}(T) \right] \leq 344(d-1)\log\left(T \right)+ 
    \left( 3546d(d-1)\beta^2 +1152(d-1)^{\frac{3}{2}}\beta \right)\log^2\left(T \right)
\end{align}

    \begin{align}
    \EX \left[ \|\theta - \hat{\theta}_t \|_2^2 \right] \leq  \frac{27648d^2\beta^2 \log ( T)+4608d\sqrt{d-1}\beta \log ( T)}{t} + \frac{4(d-1)\log (T )}{T}.
\end{align}
Using that $\beta = O (d)$ gives
\begin{align}
    \EX \left[ \textup{Regret}(T) \right] = O(d^4\log^2 (T) ), \quad  \EX \left[ \|\theta - \hat{\theta}_t \|_2^2 \right] = \tilde{O}\left( \frac{d^4}{t} \right).
\end{align}
\end{corollary}

\begin{proof}
The result of Theorem~\ref{th:regret_scaling_variance}  holds with probability at least $(1-\exp (-k/24))^{\widetilde{T}}$. Setting $k = \lceil 24\log(\widetilde{T}^2) \rceil $ gives 
    \begin{align}
        (1-\exp (-k/24))^{\widetilde{T}} \geq \left( 1 - \frac{1}{\widetilde{T}^2} \right)^{\widetilde{T}} \geq 1 - \frac{1}{\widetilde{T}}.
    \end{align}
Then given the event $R_T$ such that Algorithm~\ref{alg:linucb_vn_var} achieves the bounds given by Theorem~\ref{th:regret_scaling_variance} we have that the probability of failure is bounded by 
\begin{align}
    \mathrm{Pr}(R^C_T)  \leq \frac{1}{\widetilde{T}},
\end{align}
where we used $1 = \mathrm{Pr}(R_T)  + \mathrm{Pr}(R^C_T)  $. Then the expectation of the bad events can be bounded as
\begin{align}
    \EX \left[\text{Regret}(T)\mathbb{I}\lbrace R^C_T \rbrace \right] &\leq  4(d-1)k\widetilde{T}\mathrm{Pr}(R^C_T) \leq  4(d-1)k \\
    \EX \left[ \| \theta -\hat{\theta}_t \|_2^2 \mathbb{I}\lbrace R^C_T \rbrace \right] &\leq 4\mathrm{Pr}(R^C_T) \leq \frac{4}{\widetilde{T}} 
\end{align}
where we used $\text{Regret}(T) \leq 2T = 4(d-1)k\widetilde{T}$, $\| \theta - \hat{\theta}_t\|_2^2 \leq 4$. Finally the result follows inserting the value of $k = 24\log (\widetilde{T}^2) $ into the bounds of Theorem~\ref{th:regret_scaling_variance} and using $\widetilde{T} \leq T$.
\end{proof}

\section{Algorithm for qubit PSMAQB and numerical experiments}\label{sec:algorithm_psmaqb}

In this Section we prove our main result that is a regret bound for \textsf{LinUCB-VVN} when applied to the qubit PSMAQB problem.

\begin{figure}[h!]
    \centering
    \begin{overpic}[percent,width=0.6\textwidth]{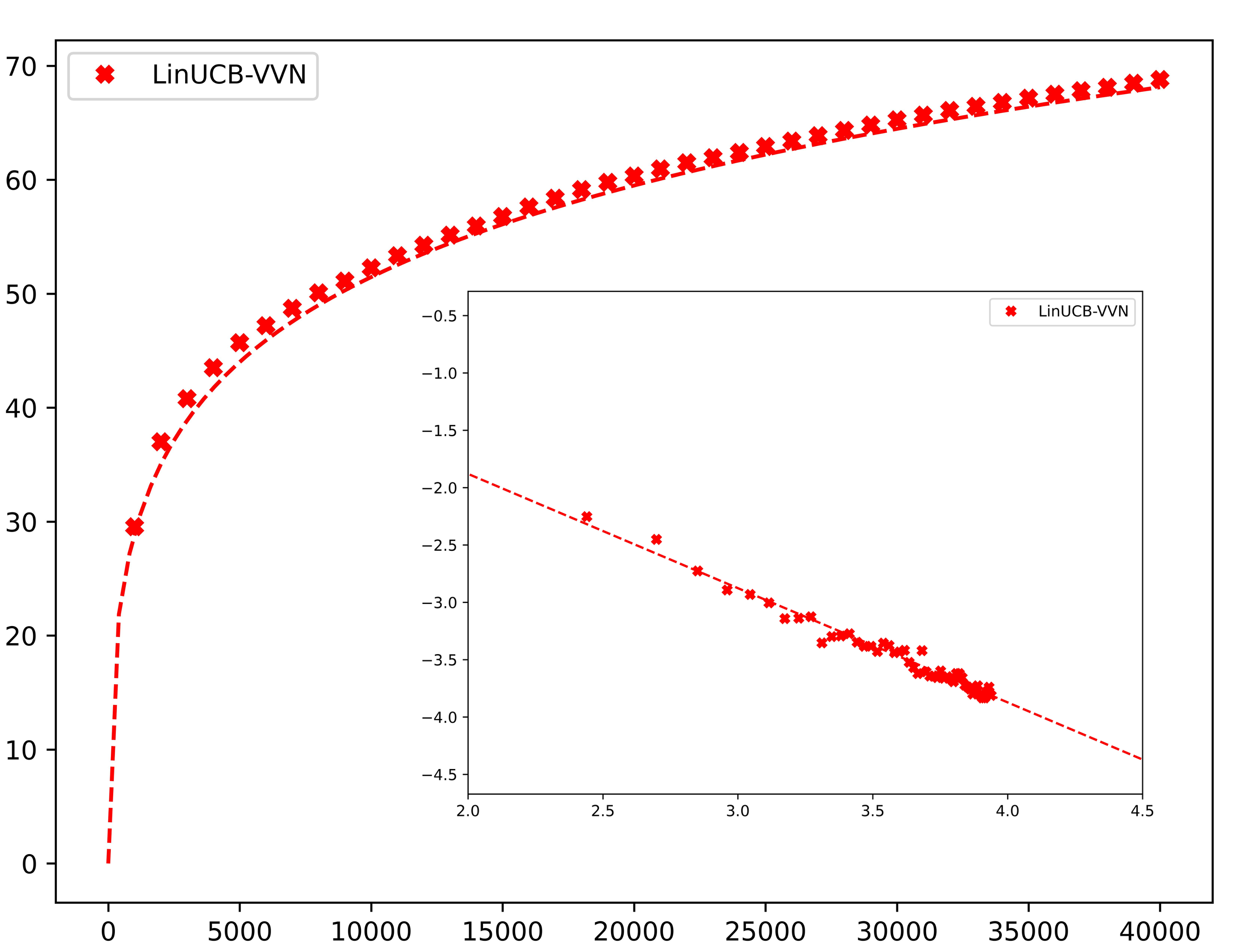}
    \put(-6,35){\rotatebox{90}{Regret($T$)}}
    \put(50,-2){$T$}
    \put(32,25){\rotatebox{90}{\tiny$\log\left( 1-F(\Pi,\Pi_t) \right)$}}
    \put(60,8){\tiny $\log\left( \frac{t}{\log (t)} \right)$}
    \end{overpic}
    \caption{Expected regret vs the number or rounds $T$ for the \textsf{LinUCB-VNN} algorithm. We run $T = 4\cdot 10^4$ rounds with $k = 10$ subsamples for the median of means construction. We use $100$ independents experiments and average over them. We obtain results for each round but only plot (red crosses) few for clarity of the figure. We fit the regression $\text{Regret}(T) = m_1\log^2 T + b_1$ with $m_1 = 3.2164 \pm 0.0009$ and $b_1 = 0.84 \pm 0.016$. In the inset plot we plot the expected infidelity  of the output estimator at each rounds $t\in [T]$ versus the number of rounds $t$. We take $\Pi_t = \Pi_{\theta^{\text{wMoM}}_t} $ as the estimator given by the median of means linear least squares estimator. We fit the regression $1- F(\Pi,\Pi_t) = b_2 \left(\frac{\log t}{t}\right)^{m_2} $ and we obtain $m_2 = -0.996 \pm 0.002$ $b_2 = 0.112\pm 0.007$. We note that the number of subsamples of the theoretical results is very conservative in comparison with the value we take for the simulations.}
    \label{fig:regret_PSMAQB}
\end{figure}

\begin{theorem}\label{th:regret_PSMAQB}
    Let $\widetilde{T}\in\mathbb{N}$ and fix $T = \lceil 96\widetilde{T}\log (\widetilde{T}^2) \rceil $. Then given a \textup{PSMAQB} with action set $\mathcal{A} = \mathcal{S}^*_2$ and environment $\Pi_{\theta}\in\mathcal{S}^*_2$ (qubits) we can apply Algorithm~\ref{alg:linucb_vn_var} for $d=3$ and it achieves
    \begin{align}
        \EX\left[\textup{Regret}(T) \right] \leq  C_1\log\left(T \right)+ 
    C_2\log^2\left(T\right).
    \end{align}
    for some universal constants $C_1, C_2 \geq 0$. Also at each time step $t\in [T ]$ it outputs an estimator  $\hat{\Pi}_{t}\in\mathcal{S}^*_2$ of $\Pi_\theta$ with infidelity scaling 
    \begin{align}
         \EX \left[ 1 -  F\left(\Pi_\theta , \hat{\Pi}_{t} \right) \right] \leq   \frac{C_3\log(T)}{t} ,
    \end{align}
    for some universal constant $C_3 \geq 0$.
\end{theorem}

\begin{proof}
    In order to apply Algorithm~\ref{alg:linucb_vn_var} to a PSMAQB we set $d=3$ (dimension for a classical linear stochastic bandit) and the actions that we select will be given by $\Pi_{a^\pm_{t,i}}$ where $a^\pm_{t,i}$ are updated as in~\eqref{eq:action_general_update}. Note that they are valid action since $ a^{\pm}_{t,i} \in\mathbb{S}^2$ imply $\Pi_{a^\pm_{t,i}} \in\mathcal{S}^*_2$. The rewards received by the algorithm follow~\eqref{eq:variance_linear_classicalquantum} with the normalization given in~\eqref{eq:classical_renormalizedreward}. This model fits into the linear bandit with linearly vanishing variance noise model explained in Section~\ref{sec:linear_bandit_vanishing_variance} and thus we can apply the guarantees established in Theorem~\ref{th:regret_scaling_variance} and Corollary~\ref{cor:expected_regret}.
    
    The algorithm is set with $k= \lceil 24\log (\widetilde{T}^2) \rceil $ batches for the MoM construction. We set $\lambda_0 = 2$, and using $\| \theta \|_2 =1 $
    we have that the constant $\beta$ given in~\eqref{eq:beta_constant} has the value
    \begin{align}
        \beta = 9\left(3\sqrt{3}+2 \right)^2 = 279+108\sqrt{3}.
    \end{align}
    Then we can check that for $d=3$ the condition~\eqref{eq:lambda_condition} for the input parameter $\lambda_0$ for Theorem~\ref{th:regret_scaling_variance} to hold is satisfied since
    \begin{align}
        \lambda_0 = 2 \geq \max\left\lbrace 2 , \frac{1}{3}+\frac{1}{2\sqrt{6}(279+108\sqrt{3})} \right\rbrace = 2 .
    \end{align}
    In the above we just substituted all numerical values. Then we are under the assumptions of Theorem~\ref{th:regret_scaling_variance} and Corollary~\ref{cor:expected_regret} and the result follows applying both results with the relation of regrets between the classical and quantum model given in~\eqref{eq:regret_relation_classicalquantum}, the relation
    \begin{align}
        \| \theta - \hat{\theta}_t \|_2^2 = 4\left(1-F\left(\Pi_\theta , \Pi_{ \hat{\theta}_t} \right) \right),
    \end{align}
    and substituting all numerical values. We take the estimator $\hat{\theta}_t$ given in Theorem~\ref{th:regret_scaling_variance} for $d=3$. We use also the bound $\widetilde{T}\leq T$ and reabsorb all the constants into $C_1,C_2,C_3$.
\end{proof}

\textbf{Remark 1.} The constant dependence can be slightly improved taking the estimator for $\Pi_\theta$ as $\Pi_{\theta^{\text{wMoM}}_t}$ with $\theta^{\text{wMoM}}_t$ defined in~\eqref{eq:action_general_update}.

\textbf{Remark 2.} The result of Theorem~\ref{th:regret_PSMAQB} also holds with high probability. In particular for the choice of batches $k = 24\log(\widetilde{T}^2)$ with probability at least $1-\frac{1}{\widetilde{T}}$.

\section{Regret lower bound for PSMAQB}\label{sec:lower_bound}

While the algorithm for PSMAQB presented above is inspired by classical bandit theory, the lower bound on the regret that we derive is essentially based on quantum information theory. The key insight here is that a policy for PSMAQB can be viewed as a sequence of state tomographies. The expected fidelity of these tomographies is linked to the regret. Hence, existing upper bounds on tomography fidelity also provide a lower bound for the expected regret of the policy. 

\subsection{Average fidelity bound for pure state tomography}
In its most general form, a tomography procedure takes $n$ copies of an unknown state $\Pi\in \mathcal{S}_d^*$ and performs a joint measurement on the state $\Pi^{\otimes n}$. This is captured in the following definition. Let $(\mathcal{S}_d^*, \Sigma)$ be a $\sigma$-algebra. A \emph{tomography scheme} is a positive operator-valued measure (POVM) $\mathcal{T}:\Sigma\to \End(\hil^{\otimes n})$ such that $\mathcal{T}(\mathcal{S}_d^*)=\Pi^+_n$, where $\Pi^+_n$ is the symmetrization operator on $\hil^{\otimes n}$. For any $\rho\in \End(\hil^{\otimes n})$, this POVM gives rise to a complex-valued measure
\begin{equation}
    P_{\mathcal{T},\rho}(A)=\Tr(\mathcal{T}(A)\rho)
\end{equation}
for $A\in\Sigma$. $P_{\mathcal{T},\rho}$ becomes a probability measure if $\rho$ satisfies $\rho\ge 0,\ \Pi^+_n\rho=\rho\Pi^+_n=\rho$, and $\Tr \rho=1$. Given $n$ copies of $\Pi$, the tomography scheme produces the distribution $P_{\mathcal{T},\Pi^{\otimes n}}$ of the predicted states. Note that $\Pi^{\otimes n}$ satisfies the properties above, so $P_{\mathcal{T},\Pi^{\otimes n}}$ is indeed a probability distribution. The fidelity of this distribution is given by
\begin{equation}
    F(\mathcal{T}, \Pi)=\int \Tr(\Pi \sigma) dP_{\mathcal{T},\Pi^{\otimes n}}(\sigma). \label{eq:avg_fidelity}
\end{equation}
Finally, the average fidelity of the tomography scheme is defined as
\begin{equation}
    F(\mathcal{T})=\int F(\mathcal{T}, |\psi\>\<\psi|)d\psi,
\end{equation}
where the integration is taken with respect to the normalized uniform measure over all pure states. In the following, $\int d\psi$ will always imply this measure. We will provide a lower bound on $F(\mathcal{T})$ in terms of $d$ and $n$, following the proof technique from~\cite{hayashi2005reexamination}. In~\cite{hayashi2005reexamination}, the proof is only presented for tomography schemes producing a finite number of predictions. For our definition, we will require more general measure-theoretic tools. Before we introduce the upper bound on the fidelity, we will prove some auxiliary lemmas about the nature of the measure $P_{\mathcal{T},\rho}$.

\begin{lemma}
  \label{lem:radon}
  Let $(\Omega, \Sigma)$ be a  $\sigma$-algebra, and let $O: \Sigma\to \End(\widetilde{\hil})$ be a POVM with values acting on a finite-dimensional Hilbert space $\widetilde{\hil}$ with $\operatorname{dim} \thil=\tdim$ s.t. $O(\Omega)\le\mathbbm{1}$, where $\mathbbm{1}$ is the identity operator. Further, let $P_{O,\sigma}: \Sigma\to\mathbb{C}$ be a complex-valued measure, defined for any $\sigma\in \End(\widetilde{\hil})$ as 
\begin{equation}
  P_{O,\sigma}(A)=\Tr[O(A)\sigma].
  \label{pdef_alt}
\end{equation}
Then, there exists a set of functions $\{f_\sigma\}$ indexed by $\sigma\in \End \widetilde{\mathcal{H}}$ that are linear w.r.t. $\sigma$ for all $\omega$ and that satisfy
\begin{equation}
  f_\sigma: \Omega\to \mathbb{C}\quad\text{s.t.}\quad \forall A\in\Sigma\ \ P_{O,\sigma}(A)=\int_A f_\sigma(\omega) dP_{O,\id}(\omega). \label{eq:radon_derivatives}
\end{equation}

\end{lemma}
We purposefully formulated this lemma with slightly more general objects than the ones used in the definition of tomography. That is, $\Omega$ does not need to be $\mathcal{S}_d^*$, and $\widetilde{\hil}$ does not need to be the n-th power $\hil^{\otimes n}$, although we will focus on this case.

\begin{proof}
Let $\{\ket{i}\}_{i=1}^{\tdim}$ be a basis of $\widetilde{\mathcal{H}}$ We will first show that $P_{O,\sigma}$ is dominated by $P_{O,\id}$ for all $\sigma$. Indeed, let $A\in \Sigma$. Assume that $P_{O,\mathbbm{1}}(A)=0$. This gives us 
\begin{equation}
  \Tr [O(A)\mathbbm{1}]=\Tr [O(A)]=0,
\end{equation} 
and, because $O(A) \ge 0$, we also have $O(A)=0$.
Therefore,
\begin{equation}
  P_{O,\sigma}(A)=\Tr[O(A)\sigma]=0.
\end{equation}
Hence, for any $\ket{i},\ket{j}$ from the basis we can introduce the Radon-Nikodym derivatives $f_{\ket{i}\bra{j}}$, which will satisfy~\eqref{eq:radon_derivatives}. Then, for any $\sigma\in \End \widetilde{\mathcal{H}}$ we can define
\begin{equation}
    f_\sigma(\omega)=\sum_{i,j=1}^{\tdim} \bra{i}\sigma\ket{j}f_{\ket{i}\bra{j}}(\omega).
\end{equation}
These $f_\sigma$ are linear in $\sigma$ by definition. A direct calculation shows that they also satisfy~\eqref{eq:radon_derivatives}.
\end{proof}

Note that for $\sigma\ge 0$, the measure $P_{O,\sigma}$ is finite and nonnegative, but nonnegativity (and even real-valuedness) do not hold for a general $\sigma\in \End(\widetilde{\hil})$. 

By our definition of $f_\sigma(\omega)$, it can be written as
\begin{equation}
    f_\sigma(\omega)=\Tr\left[K(\omega)\sigma\right],\quad\text{where}\ K(\omega)=\sum_{i,j=1}^{\tdim}f_{|i\>\<j|}(\omega)|j\>\<i|. \label{eq:introk}
\end{equation}
As the following lemma demonstrates, $K(\omega)\ge 0$ for $P_{O,\id}$-almost every $\omega$:

\begin{lemma}
  \label{lem:magic}
  Let $(\Omega, \Sigma, \mu)$ be a measurable space and $V: \Omega\to \End(\thil)$ be a measurable operator-valued function with values acting on a finite-dimensional Hilbert space $\thil$ such that
\begin{equation}
  \forall A\in\Sigma\quad \int_A V(\omega)d\mu(\omega) \ge 0.
\end{equation}
Then, $V(\omega) \ge 0$ $\mu$-almost everywhere.
\end{lemma}

\begin{proof}
  Let $\ket{\psi}\in \thil$ and define
\begin{equation}
  g_{\psi}(\omega) = \bra{\psi}V(\omega)\ket{\psi}.
\end{equation}
By the given condition, for any $A\in\Sigma$
\begin{equation}
  \int_A g_{\psi}(\omega)d\mu(\omega)= \bra{\psi}\int_A V(\omega)d\mu(\omega) \ket{\psi} \ge 0.
\end{equation}
It follows that $g_{\psi}(\omega)\ge 0$ $\mu$-almost everywhere. Let 
\begin{equation}
  Z_\psi=\{\omega\in \Omega \text{ s.t. } g_\psi(\omega) <0\}
\end{equation}
We have shown that $\mu(Z_\psi)=0$. Next, since  $\thil$ is finite-dimensional, it is separable. Therefore, there exists a countable set $\{\ket{\psi_k}\}_k$ dense in $\thil$. Let
\begin{equation}
  Z=\bigcup_k Z_{\psi_k}.
\end{equation}
We have that $\mu(Z)=0$. Finally, let $\omega\in \Omega \setminus Z$ and $\ket{\psi}\in \thil$. Because $\{\ket{\psi_k}\}$ is dense in $\thil$, there exists a sequence $\{\ket{\psi_{k_i}}\}$ converging to $\ket{\psi}$. Then,
\begin{equation}
  0 \le \bra{\psi_{k_i}}V(\omega)\ket{\psi_{k_i}} \xrightarrow{i \to \infty } \bra{\psi}V(\omega)\ket{\psi}.
\end{equation}
Overall, we get that
\begin{equation}
\forall \omega \in \Omega \setminus Z,\ \ket{\psi}\in \thil\quad \bra{\psi}V(\omega)\ket{\psi} \ge 0.
\end{equation}
Together with $\mu(Z)=0$, this gives the desired result.
\end{proof}

Now we can apply this analysis to the POVM corresponding to our tomography scheme, and get the desired upper bound on the fidelity.

\begin{theorem}\label{th:fidelity_upper_bound}
  For any tomography scheme $\mathcal{T}$ utilizing $n$ copies of the input state, the average fidelity is bounded by
\begin{equation}
  F(\mathcal{T})\le \frac{n+1}{n+d}.
\end{equation}
\end{theorem}

\begin{proof}
  We will introduce the density $K(\omega)$ from~\eqref{eq:introk} for our tomography scheme $\mathcal{T}$ and the corresponding measure $P_{\mathcal{T},\sigma}$. Lemma~\ref{lem:radon} allows us to introduce for any $\sigma\in\End(\hil^{\otimes n})$ the density $f_\sigma:\Omega\to\mathbb{C}$ s.t.
\begin{equation}
 \forall A\in\Sigma\ \ P_{\mathcal{T},\sigma}(A)=\int_A f_\sigma(\omega) dP_{\mathcal{T},\mathbbm{1}}(\omega).
\end{equation}
This density can be written as $f_\sigma(\omega)=\Tr\left(K(\omega)\sigma\right)$ 
for some $K(\omega)\in\End(\hil^{\otimes n})$. $K(\omega)$ can be considered as the operator-valued density of $\mathcal{T}$ w.r.t. $P_{\mathcal{T},\id}$:
\begin{equation}\label{eq:opdensity}
    \forall A\in\Sigma\quad \mathcal{T}(A)=\int_A K(\omega)dP_{\mathcal{T},\id}(\omega).
\end{equation}
Since $\mathcal{T}(A)\ge 0$, it follows by Lemma~\ref{lem:magic} that $K(\omega)\ge 0$ for $P_{\mathcal{T},\id}$-almost all $\omega$. Furthermore, as $\mathcal{T}(\mathcal{S}_d^*)=\Pi^+_n$, we have that for all $A\in\Sigma$,
$\mathcal{T}(A)\le \Pi^+_n$. Therefore, $\mathcal{T}(A)\Pi^+_n=\Pi^+_n \mathcal{T}(A)=\mathcal{T}(A)$. This means that $\tilde{K}(\omega)=\Pi^+_nK(\omega)\Pi^+_n$ would also satisfy~\eqref{eq:opdensity}. In the following, we will without loss of generality assume that
\begin{equation}
    K(\omega)=\Pi^+_nK(\omega)=K(\omega)\Pi^+_n.
\end{equation}

With these tools at hand, we are ready to adapt the proof from~\cite{hayashi2005reexamination} to the general case of POVM tomography schemes. We begin by rewriting the expression~\eqref{eq:avg_fidelity} for average fidelity:
\begin{align}
  F(\mathcal{T}) &= \int d\psi \int dP_{\mathcal{T},(\ket{\psi}\!\bra{\psi})^{\otimes n}}(\sigma) \Tr(\sigma\ket{\psi}\!\bra{\psi}) \\
                 &= \int d\psi \int dP_{\mathcal{T},\id}(\sigma)\Tr(\ket{\psi}\!\bra{\psi}\sigma)\Tr\left( K(\sigma)(\ket{\psi}\!\bra{\psi})^{\otimes n} \right). \label{eq:fidelity1}
\end{align}
Since fidelity is nonnegative and its average is bounded by 1, we can change the order of integration. Following~\cite{hayashi2005reexamination}, we introduce notation
\begin{equation}
  \sigma_n(k)=\mathbbm{1}^{\otimes (k-1)}\otimes \sigma \otimes \mathbbm{1}^{\otimes (n-k)}\in \hil^{\otimes n}.
\end{equation}
The product of traces in~\eqref{eq:fidelity1} can be rewritten in the following manner:
\begin{align}
\begin{split}
    F(\mathcal{T})& =\int dP_{\mathcal{T},\id}(\sigma)\int d\psi \\
    & \quad \quad \Tr\left( (K(\sigma)\otimes \mathbbm{1})(\ket{\psi}\!\bra{\psi})^{\otimes (n+1)}\sigma_{n+1}(n+1) \right). \label{eq:fidelity2}
\end{split}
\end{align}
We can now take the inner integral in closed form. As shown in~\cite[Eq.~(4)]{hayashi2005reexamination},
\begin{equation}
  \int d\psi (\ket{\psi}\!\bra{\psi})^{\otimes n}=\frac{\Pi^+_n}{D_n},
  \label{eq:avgpsi}
\end{equation}
where $D_n=\binom{n+d-1}{d}$. Another useful result in this paper is~\cite[Eq.~(8)]{hayashi2005reexamination}:
\begin{equation}
  \Tr_{n+1}\left( \Pi^+_{n+1}\sigma_{n+1}(n+1) \right)=\frac{1}{n+1}\Pi^+_n\left(\mathbbm{1}+\sum_{k=1}^n\sigma_n(k)\right),
  \label{eq:trs}
\end{equation}
where $\Tr_{n+1}:\End(\hil^{\otimes (n+1)})\to\End(\hil^{\otimes n})$ is the partial trace on the $(n+1)$-st copy of the system. These 
expressions allow us to rewrite~\eqref{eq:fidelity2} as follows:
\begin{align}
  F(\mathcal{T}) &= \frac{1}{D_{n+1}}\int dP_{\mathcal{T},\id}(\sigma)\Tr\left( (K(\sigma)\otimes \mathbbm{1})\Pi^+_{n+1} \sigma_{n+1}(n+1)  \right) \label{eq:avgnp1} \\
                 &= \frac{1}{D_{n+1}}\int dP_{\mathcal{T},\id}(\sigma)\Tr\left( K(\sigma)\Tr_{n+1}\left(\Pi^+_{n+1} \sigma_{n+1}(n+1)\right)\right)  \\
                 &= \frac{1}{(n+1)D_{n+1}}\int dP_{\mathcal{T},\id}(\sigma)\Tr\left( K(\sigma)\left(\mathbbm{1}+\sum_{k=1}^n\sigma_n(k)  \right)\right).
\end{align}
Finally, $\sigma_n(k)\le\id$, so $\Tr(K(\sigma)\sigma_n(k))\le \Tr(K(\sigma))$, and we can bound the above as
\begin{align}
F(\mathcal{T})\le \frac{1}{D_{n+1}}\int dP_{\mathcal{T},\id}(\sigma)\Tr\left( K(\sigma)\right) \nonumber \\ = \frac{\Tr \Pi^+_n}{D_{n+1}}=\frac{D_n}{D_{n+1}}=\frac{n+1}{n+d}.
\end{align}
\end{proof}

\subsection{Bandit policy as a sequence of tomographies}
\begin{theorem}
  Given a $d$-dimensional pure state general multi-armed quantum bandit we have that for any policy $\pi$ the average expected regret is bounded by
  \begin{align}
    \int d\psi \EX_{\ket{\psi}\!\bra{\psi},\pi} \left[\textup{Regret}(T,\pi,\ket{\psi}\!\bra{\psi}) \right] \geq (d-1)\log\left(\frac{T}{d+1} \right),
  \end{align}
where the expectation is taken w.r.t. the measure~\eqref{eq:prob_measure_maqb} over actions taken by the bandit, and the regret is defined in~\eqref{eq:def_regret}.
\end{theorem}
The above Theorem gives $\EX \left[ \text{Regret}(T) \right] =\Omega(d\log \frac{T}{d})$. In the case of qubit environments, we have $d=2$ and ${\EX \left[ \text{Regret}(T) \right]=\Omega(\log T)}$.

\begin{proof}
  Given a policy $\pi$, we can introduce a POVM $E_t: (\Sigma\times \{0,1\})^{\times t}\to\End(\hil^{\otimes t})$ such that
  \begin{equation}
    P^t_{\ket{\psi}\!\bra{\psi}, \pi}(A_1,r_1,\dotsc,A_t,r_t)=\Tr\left( (\ket{\psi}\!\bra{\psi})^{\otimes t}E_t(A_1,r_1,\dotsc,A_t,r_t) \right),
  \end{equation}
where $P^t_{\ket{\psi}\!\bra{\psi},\pi}$ is the probability measure defined by~\eqref{eq:prob_measure_maqb}, but only for actions and rewards until step $t$. The construction of this POVM is presented in the proof of Lemma~9 in~\cite{lumbreras22bandit}.
We will also define the coordinate mapping
\begin{equation}
  \Psi_t(\Pi_{1}, r_1,\dotsc,\Pi_{t},r_t)=\Pi_{t},
\end{equation}
where $\Pi_{i}\in\mathcal{A}$ are actions and $r_i\in\{0,1\}$ are rewards of the PSMAQB. Now we can for each step $t$ define a tomography scheme $\mathcal{T}_t=E_t\circ \Psi_t^{-1}$ as the pushforward POVM from $E_t$ to the space $(\mathcal{A},\Sigma)$. Informally, this tomography scheme takes $t$ copies of the state, runs the policy $\pi$ on them, and outputs the $t$-th action of the policy as the predicted state. For $A\in\Sigma$, we can rewrite the tomography's distribution on predictions as
\begin{align}
    &P_{\mathcal{T},(\ket{\psi}\!\bra{\psi})^{\otimes t}}(A)=\Tr\left(\mathcal{T}_t(A)(\ket{\psi}\!\bra{\psi})^{\otimes t}\right)\nonumber \\&=\Tr\left(E_t(\Psi_t^{-1}(A))(\ket{\psi}\!\bra{\psi})^{\otimes t}\right)=\left(P^t_{\ket{\psi}\!\bra{\psi},\pi}\circ \Psi^{-1}\right)(A).
\end{align}
Then, the fidelity of $\mathcal{T}_t$ on the input $\ket{\psi}\!\bra{\psi}$ can be rewritten as
\begin{align}
    &F(\mathcal{T}_t,\ket{\psi}\!\bra{\psi}) =\int \<\psi|\rho|\psi\> dP_{\mathcal{T}_t,(\ket{\psi}\!\bra{\psi})^{\otimes t}}(\rho) \\
    &= \int \<\psi|\Psi_t(\Pi_1,r_1,\cdots,\Pi_t,r_t)|\psi\> dP^t_{\ket{\psi}\!\bra{\psi},\pi}(\Pi_1,r_1,\cdots,\Pi_t,r_t) \\
    &= \EX_{\ket{\psi}\!\bra{\psi},\pi}\left[\<\psi|\Pi_{t}|\psi\>\right].
\end{align}
Using the bound for average tomography fidelity on $\mathcal{T}_t$ from Theorem~\ref{th:fidelity_upper_bound}, we can now bound the average regret of $\pi$:
\begin{align}
  &\int \EX_{\ket{\psi}\!\bra{\psi}}\left[\textup{Regret}(T, \pi,\ket{\psi}\!\bra{\psi})\right]d\psi \\&= T-\sum_{t=1}^T\int \ \mathbb{E}_{\ket{\psi}\!\bra{\psi}}\left[\<\psi|\Pi_{t}|\psi\>\right]d\psi \label{eq:sumint}
  \end{align}

  \begin{align}
  &= T-\sum_{t=1}^T F(\mathcal{T}_t) \label{eq:introt} \geq \sum_{t=1}^T 1 - \frac{t+1}{t+d} \\
  & = \sum_{t=1}^T \frac{d-1}{t+d} \geq (d-1)\log\left( \frac{T}{d+1} \right),
\end{align}
where the last inequality follows from bounding the sum with the integral of the function $f(t) = 1/(t+d)$.
\end{proof}

\paragraph*{Acknowledgements:} JL thanks Jan Seyfried and Yanglin Hu for comments and suggestions, Erkka Happasalo for discussions about disturbance and Roberto Rubboli for many technical discussions. 
Mikhail Terekhov is grateful to be supported by the EDIC Fellowship from the School of Computer Science at EPFL. 
Josep Lumbreras and Marco Tomammichel are supported by the National Research Foundation, Singapore and A*STAR under its CQT Bridging Grant and the Quantum Engineering Programme grant NRF2021-QEP2-02-P05.

\bibliographystyle{ultimate}
\bibliography{biblio_purestatebandit}

\end{document}